\documentclass{amsart}
\usepackage{times}
\usepackage[dvips]{color}
\usepackage{amssymb,amsmath,amscd,amsthm,stmaryrd,verbatim, mathrsfs}

\newtheorem{Thm}{Theorem}
\newtheorem{theorem}{Theorem}[section]

\newtheorem{Th}{Theorem}[section]
\newtheorem{Lem}[theorem]{Lemma}

\newtheorem{Prop}{Proposition}[section]
\newtheorem{cor}{Corollary}

\newtheorem{Definition}{Definition}
\newtheorem{Def}{Definition}[section]

\newtheorem{rmk}{Remark}[section]

\numberwithin{equation}{section}

\newcommand{\bb}{\mathbb}
\newcommand{\ms}{\mathscr}
\newcommand{\mr}{\mathrm}
\newcommand{\frk}{\mathfrak}
\newcommand{\tr}{\mr{tr}\,}

\usepackage{hyperref}
\begin{document}

\title{Euclidean Jordan Algebras, Hidden Actions, and J-Kepler Problems}
\author{Guowu Meng}

\address{Department of Mathematics, Hong Kong Univ. of Sci. and
Tech., Clear Water Bay, Kowloon, Hong Kong}
\curraddr{Institute for Advanced Study, Einstein Drive, Princeton, New Jersey 08540 USA }

\email{mameng@ust.hk}
\thanks{The author was supported in part by Qiu Shi Science and Technologies Foundation
and an internal grant (``Scholol-based Initiatives" ) from the Hong Kong Univ. of Sci. \& Tech.}


\date{May 11, 2011}



\begin{abstract} For a {\em simple  euclidean Jordan algebra}, let
$\frk{co}$ be its conformal algebra, $\ms P$ be the manifold
consisting of its semi-positive rank-one elements, $C^\infty(\ms P)$ be the space of complex-valued smooth functions on $\ms P$.
An explicit action of $\frk{co}$ on $C^\infty(\ms P)$, referred to as the {\em hidden
action} of $\frk{co}$ on $\ms P$, is exhibited. This hidden action turns out to be mathematically responsible for the existence of the Kepler problem and its recently-discovered vast generalizations, referred to as J-Kepler problems. The J-Kepler problems are then reconstructed and re-examined in terms of the unified language of euclidean Jordan algebras. As a result, for a
simple euclidean Jordan algebra, the minimal representation of its
conformal group can be realized either as the Hilbert space of bound
states for its J-Kepler problem or as $L^2({\ms P}, {1\over
r}\mr{vol})$, where $\mr{vol}$ is the volume form on $\ms P$ and
$r$ is the inner product of $x\in \ms P$ with the identity
element of the Jordan algebra.

\end{abstract}

 \maketitle

\section {Introduction}
The main message we wish to convey in this paper is that the simple
euclidean Jordan algebras introduced by P. Jordan \cite{PJordan33} in
the 1930's and the various Kepler-type problems we \cite{meng} introduced
in recent years are {\em intrinsically} in one-to-one
correspondence: for a simple euclidean Jordan algebra, there is a
super integrable model whose configuration space is the manifold
consisting of the semi-positive rank-one elements of the Jordan algebra;
moreover, the conformal symmetry group of the Jordan algebra is the dynamical symmetry group of the super
integrable model. Since they resemble the Kepler/Coulomb problem, these
super integrable models are referred to as {\bf J-Kepler problems}.

\vskip 10pt

The euclidean Jordan algebras were initially introduced by P. Jordan
for the purpose of reformulating quantum mechanics in a minimal way.
By definition, an {\bf euclidean Jordan algebra} is {\it a finite
dimensional real commutative algebra $V$ with unit such that, for
$a$, $b$ in $V$, 1) $a^2(ab)=a(a^2b)$, 2) $a^2+b^2=0$ implies that $a=b=0$}. As an example, we have the euclidean Jordan
algebra of real symmetric $k\times k$-matrices with the Jordan
product being the symmetrized matrix product. A theorem of
Jordan, von Neumann and Wigner \cite{JVW34} says that the simple
euclidean Jordan algebras consist of {\it four infinity families and
one exceptional}.

Although physicists quickly lost interest in Jordan algebras, the
subsequent further explorations taken on by mathematicians turned out
to be quite fruitful. The work of the Koecher school is especially
relevant to the Kepler problem. Here is an important discovery made
by M. Koecher \cite{Koecher99, FK91}: {\it simple euclidean Jordan algebras and
(irreducible) tube domains are in natural one-to-one correspondence}.
Combining with our discovery, we conclude that {\it J-Kepler
problems and (irreducible) tube domains are in natural one-to-one correspondence.}

\vskip 10pt The {\it {\bf Kepler problem} is a physics problem about
two bodies which attract each other by a force inversely
proportional to the distance}. Mathematically, this is a mechanical
system with configuration space ${\bb R}^3_*:=\bb R^3\setminus\{0\}$
and Lagrangian
$$
L={1\over 2}{\bf r}'\cdot {\bf r}' +{1\over r}.
$$ Here, $\bf r$ is a function of time $t$ taking value in $\bb R^3_*$, $r=|\bf r|$ and $\bf r'$ is the time derivative of $\bf r$.
Therefore, quantum mechanically the hamiltonian for the Kepler
problem becomes
\begin{eqnarray}\label{hydroH}
\hat H=-{1\over 2}\Delta-{1\over r}.
\end{eqnarray}Here, $\Delta$ is the Laplace operator on $\bb R^3_*$. Physicists are interested in solving the bound
state eigenvalue problems for $\hat H$, i.e., 1) finding the list of
real numbers $\lambda_0<\lambda_1<\cdots$ such that
$$
\ms H_\lambda:=\{\psi\in C^\infty({\bb R}_*^3, {\bb C})\mid \int
_{\bb R_*^3}|\psi|^2\, d^3\vec r<\infty, \hat H\psi=\lambda\psi \}
$$ is nontrivial if and only if $\lambda$ is one of $\lambda_I$; 2)
determining each $\ms H_{\lambda_I}$. It is well known that, for the
hamiltonian in Eq. (\ref{hydroH}), the answer is very simple:
$$
\lambda_I=-{1/2\over (I+1)^2}, \quad I=0, 1, \ldots
$$
and as a module of $\frk{so}(4)=\frk{so}(3)\oplus \frk{so}(3)$, $\ms H_I\cong V_I\otimes V_I$, where $V_I$ is the $\frk{so}(3)$-module with dimension $I+1$. What is
less well known is a discovery of A.O. Barut and H. Kleinert \cite{BarutKleinert67} which essentially says that
$$
{\mathcal H}:=\bigoplus_{I=0}^\infty {\ms H}_I
$$ is a unitary lowest weight $(\frk{so}(6), \mr{SO}(4)\times \mr{SO}(2))$-module in the sense of
Harish-Chandre \cite{Harish-Chandra1953}, hence it can be integrated to a unitary lowest
weight module for $\mr{SO}_0(2, 4)$ (the identity component of $\mr{SO}(2,
4)$). Moreover, this module for $\mr{SO}_0(2, 4)$ is minimal in the sense
of A. Joesph \cite{Joseph1974} and also has $L^2({\bb R}_*^3, {1\over r}d^3\vec r)$ as a
geometric realization.

An apparent mathematical generalization, known to many people, is to
replace $\bb R^3$ by $\bb R^n$ ($n\ge 2$) and keep the hamiltonian
in the same form. Similar results are valid: $\lambda_I=-{1/2\over
(I+{n-1\over 2})^2}$, $\ms H_I$ is an irreducible representation of
$\mr{SO}(n+1)$, $\bigoplus_{I=0}^\infty {\ms H}_I$ is a unitary lowest
weight $(\frk{so}(n+3), \mr{SO}(n+1)\times \mr{SO}(2))$-module, hence it can
be integrated to a unitary lowest weight module for $\mr{SO}_0(2, n+1)$
(actually a double of it when $n$ is even); moreover, this module is
minimal and also has $L^2({\bb R}_*^n, {1\over r}d^n\vec r)$ as a
geometric realization.

The less obvious cousins of the Kepler problems were all worked out in recent years. Together
with the obvious ones mentioned in the preceding paragraphs, they
consist of {\it four infinity families and one exceptional}. So it
appears that there is a natural one-to-one correspondence between
them and the simple euclidean Jordan algebras.

Indeed, this is the case, and here is a quick way to see the one-to-one
correspondence. We begin with the notion of {\bf Kepler cone} for a
simple euclidean Jordan algebra. By definition it is the manifold
consisting of the rank-one semi-positive elements of the Jordan algebra,
equipped with a suitable Riemannian metric. This Kepler cone plays
the role of $\bb R^3_*$. To define the {\bf J-Kepler problem}, one
needs to replace the hamiltonian in Eq. (\ref{hydroH}) by this one:
\begin{eqnarray}
\hat h=-{1\over 2}\Delta-\left( {B\over 2r^2}+ {1\over r}\right).
\end{eqnarray}
Here, $\Delta$ is the (non-positive) Laplace operator on the Kepler
cone, $r=r(x)$ is the inner product of $x$ (in the Kepler cone) with
the identity element of the Jordan algebra, and $B$ is a constant
depending on the Jordan algebra, for example, $B=26$ for the
exceptional Jordan algebra. Note that, when the Jordan algebra is
the Minkowski space, the
J-Kepler problem is equivalent to the Kepler problem, i.e., it is a reformulation of
the Kepler problem as a dynamic problem on the {\it open future light cone}:
$$
\{(t, {\bf x})\mid t^2=|{\bf x}|^2, t>0\}
$$
--- the Kepler cone in this case.

The J-Kepler problems all share the key features of the Kepler problem, for example, the $I$-th eigenvalue of the
hamiltonian is
$$
\lambda_I=-{1/2\over (I+\rho \delta/4)^2},
$$ here $\rho$ and $\delta$ are the rank and degree of the Jordan algebra respectively. The more detailed results are given in Theorem \ref{main3}
on page \pageref{main3}.

A key mathematical result here is Theorem \ref{Main1} on page \pageref{Main1}, which gives
an explicit action of the conformal algebra of the Jordan algebra on
the space of complex-valued smooth functions on the Kepler cone. In
this action, elements of the conformal algebra are realized as
differential operators of degree zero, one and two, so this action
does not come from an underlying action on the Kepler cone;
consequently such an action is referred to as a {\bf hidden action} of
the conformal algebra on the Kepler cone. In our view, this hidden
action is the mathematical origin for the J-Kepler problems.

One curious fact, though not presented here, is that the magnetized
versions of a given J-Kepler problem also exist unless the Jordan
algebra is the exceptional one. This could lead to some speculations about
the fundamental physics. Another fact, which is not proved here either, is that a unitary lowest weight representation can be realized by the Hilbert space of bound states of a magnetized version of a J-Kepler problem if and
only if it has the minimal positive Gelfand-Kirillov dimension.

\vskip 10pt Here is the organization of this paper. In section
\ref{S:JA}, we give a review of the euclidean Jordan algebras,
tailed to our needs. In section \ref{S:TKK}, we review the TKK
(Tits-Kantor-Koecher) construction \cite{TKK60}, something that assigns a simple real Lie algebra (the {\bf
conformal algebra}) to each simple euclidean Jordan algebra. In
section \ref{S:Vogan}, we do a bit of structural analysis for the
conformal algebra. In section \ref{S:KM}, we introduce the notion
of Kepler cone for simple euclidean Jordan algebras.  In
section \ref{S:DS}, we introduce the hidden action of the
conformal algebra on the Kepler cone, which amounts to the dynamic
symmetry of the corresponding J-Kepler problem. In section
\ref{S:KP}, we introduce the notion of J-Kepler problem
for simple euclidean Jordan algebras and show that a J-Kepler
problem is just a Kepler-type problem with zero magnetic charge
that we introduced and studied in recent years. In section
\ref{S:Summary}, based on the hidden action obtained from section
\ref{S:DS}, we give the dynamical symmetry analysis and solve the
bound state problem for the J-Kepler problem.
\vskip 10pt
{\bf Acknowledgement}. We wish to thank the following distinguished mathematicians for providing their moral support for our lonely exploration: M. Atiyah, R. Howe, J.S. Li and C. Taubes.

\section{euclidean Jordan Algebras}\label{S:JA}
The materials reviewed in this and the next sections can be found in Refs. \cite{Koecher99, FK91}.
Recall that an {\bf algebra} $V$ over a field $\bb F$ is a vector
space over $\bb F$ together with a $\bb F$-bilinear map $V\times
V\to V$ which maps $(u, v)$ to $uv$. This $\bb F$-bilinear map can
be recast as a linear map $V\to \mr{End}_{\bb F}(V)$ which maps $u$
to $L_u$: $v\mapsto uv$.

We say that algebra $V$ is {\it commutative} if $uv=vu$ for any $u,
v\in V$. As usual, we write $u^2$ for $uu$ and $u^{m+1}$ for $uu^m$
inductively.

\begin{Def}A {\bf Jordan algebra} over $\bb F$ is just a
commutative algebra $V$ over $\bb F$ such that
\begin{eqnarray}\label{JA}
[L_u, L_{u^2}]=0
\end{eqnarray} for any $u\in V$.
\end{Def}
Here are some {\it basic facts} about the Jordan algebra $V$ over $\bb
F$:
\begin{itemize}

\item  $u^r u^s=u^{r+s}$ for any $u\in V$ and
any integers $r,s\ge 1$.

\item $L_{u^m}$ ($m\ge 1$) is a polynomial in $L_u$
and $L_{u^2}$, for example,
\begin{eqnarray}\label{ID:1}
L_{u^3}=3L_{u^2}L_u-2L_u^3.
\end{eqnarray}

\item For any $u, v, z\in V$, we have
\begin{eqnarray}\label{usefulid}
[[L_u, L_v], L_z]=L_{u(vz)-v(uz)}
\end{eqnarray}
provided that $\mbox{Char}(\bb F)\neq 2$.

\end{itemize}

As the first example, we note that $\bb F$ is a Jordan algebra over
$\bb F$. Here is a {\it recipe} to produce Jordan algebras. Suppose
that $\Phi$ is an associative algebra over field $\bb F$ with
characteristic $\neq 2$, and $V\subset \Phi$ is a linear subspace of
$\Phi$, closed under square operation, i.e, $u\in V$ $\Rightarrow$
$u^2\in V$. Then $V$ is a Jordan algebra over $\bb F$ under the
Jordan product
$$
uv:={(u+v)^2-u^2-v^2\over 2}.
$$
Applying this recipe, we have the following Jordan algebras over
$\bb R$:
\begin{enumerate}
\item The algebra $\Gamma(n)$. Here $\Phi=\mr{Cl}(\bb R^n)$---the Clifford algebra of $\bb
R^n$ and $V=\bb R\oplus \bb R^n$.

\item The algebra ${\mathcal H}_n(\bb R)$. Here $\Phi=M_n(\bb R)$---the algebra of real $n\times n$-matrices and $V\subset \Phi$ is the set of symmetric
$n\times n$-matrices.

\item The algebra ${\mathcal H}_n(\bb C)$. Here $\Phi=M_n(\bb C)$---the algebra of complex $n\times n$-matrices (considered as an algebra over $\bb R$)
and $V\subset \Phi$ is the set of Hermitian $n\times n$-matrices.

\item The algebra ${\mathcal H}_n(\bb H)$. Here $\Phi=M_n(\bb H)$---the algebra of quaternionic $n\times n$-matrices (considered as an algebra over $\bb R$)
and $V\subset \Phi$ is the set of Hermitian $n\times n$-matrices.

\end{enumerate}
The Jordan algebras over $\bb R$ listed above are {\bf
special} because they can be derived from associated
algebras via the above recipe. Let us denote by ${\mathcal H}_n(\bb O)$ the algebra for which
the underlying real vector space is the set of Hermitian
$n\times n$-matrices over octonions $\bb O$ and the  product is the
symmetrized matrix product. One can show that ${\mathcal
H}_n(\bb O)$ is a Jordan algebra if and only if $n\le 3$.

\vskip 10pt Any Jordan algebra $V$ comes with a canonical symmetric
bilinear form \begin{eqnarray}\tau(u, v):=\mbox{the trace of
$L_{uv}$}.\end{eqnarray} Thanks to Eq. (\ref{usefulid}), $L_u$ is self-adjoint with respect to $\tau$.

We say that the Jordan algebra $V$ is {\it semi-simple} if the symmetric
bilinear form $\tau$ is non-degenerate. We say that the Jordan algebra
$V$ is {\it simple} if it is semi-simple and has no ideal other than
$\{0\}$ and $V$ itself.

By definition, an {\bf euclidean Jordan algebra}\footnote{Called
{\it formally real Jordan algebra} in the old literature.} is a
real Jordan algebra with an identity element and the positive definite symmetric bilinear form $\tau$. Therefore, an euclidean Jordan algebra is
semi-simple and can be uniquely written as the direct sum of simple
ideals --- ideals which are simple as Jordan algebras.

\begin{Th}[Jordan, von Neumann and Wigner]\label{JAclassification}
The complete list of simple euclidean Jordan algebras are
\begin{enumerate}
\item The algebra $\Gamma(n)=\bb R\oplus \bb R^n$ ($n\ge 2$).

\item The algebra ${\mathcal H}_n(\bb R)$ ($n\ge 3$ or $n=1$).

\item The algebra ${\mathcal H}_n(\bb C)$ ($n\ge 3$).

\item The algebra ${\mathcal H}_n(\bb H)$ ($n\ge 3$).

\item The algebra ${\mathcal H}_3(\bb O)$.
\end{enumerate}
\end{Th}
Note that the last one is called the {\bf exceptional type} because it
cannot be obtained from an associative algebra via the recipe we mentioned early. Note also that
$\Gamma(1)$ is not simple, ${\mathcal H}_1(\bb F)=\bb R$ is the
only associative simple euclidean Jordan algebra, and we have
various isomorphisms: $\Gamma(2)\cong {\mathcal H}_2(\bb R)$,
$\Gamma(3)\cong {\mathcal H}_2(\bb C)$, $\Gamma(5)\cong {\mathcal
H}_2(\bb H)$, $\Gamma(9)\cong {\mathcal H}_2(\bb O)$. The first family is called the {\bf Dirac type} and the next three families are called the {\bf hermitian type}.

\vskip 10pt Throughout the remainder of this paper, we always use
$V$ to denote a simple euclidean Jordan algebra, and $e$ to denote
the identity element of $V$.

The notion of {\bf trace} is valid for Jordan algebras. For the
simple euclidean Jordan algebras, the trace can be easily described:
For $\Gamma(n)$, we have
$$
\tr(\lambda, \vec u)=2\lambda,
$$
and for all other types, it is the usual
trace on matrices.

Recall that $L_u$: $V\to V$ is the multiplication by $u$ and its
trace is denoted by $\mr{Tr}(L_u)$. It is a fact that $${1\over \dim
V}\mr{Tr}(L_u)={1\over \rho}\tr (u)$$ where $\rho:=\tr(e)$ is the {\bf
rank} of $V$.

For the {\bf inner product} on $V$, we take
\begin{eqnarray}\fbox{$\langle u\mid v\rangle :={1\over
\rho}\tr(uv)$}\end{eqnarray} so that $e$ becomes a unit vector. One
can check that $L_u$ is self-adjoint with respect to this inner
product: $ \langle v u\mid w\rangle = \langle v\mid u w\rangle$,
i.e., $L_u'=L_u$. 

We shall use Dirac's bracket notations: for $u, v\in V$, we
declare that $\mid u\rangle$ is the vector $u$, $\langle u\mid$
is the co-vector sending $z\in V$ to $\langle u\mid z\rangle$, and $\mid u\rangle\langle v\mid$ is the endomorphism
sending $z\in V$ to $\langle v\mid z\rangle u$. Since $L_u'=L_u$, the dual of $L_u$, denoted by $L_u^*$, maps $\langle z\mid$ to $-\langle L_u z\mid$.

We shall adopt the following conventions: $x$ is reserved for a point in
the smooth space $V$, and $u$, $v$, $z$, $w$ are reserved for vectors in
the vector space $V$. We shall fix an orthonomal basis $\{e_\alpha\}$ for $V$, with respect to which, we can express $x$ as $\sum_\alpha x^\alpha e_\alpha$. For simplicity, we shall write $\sum_\alpha e_\alpha {\partial \over \partial x^\alpha}$ as ${/\hskip -5pt
\partial}$.

From here on, we shall also use $V$ to denote the euclidean space
with underlying smooth space $V$ and Riemannian metric $ds^2_E$:
\begin{eqnarray}
T_xV\times T_xV & \to &  \bb R \cr ((x, u), (x, v)) &\mapsto &
\langle u\mid v\rangle.
\end{eqnarray}

For $u\in V$, we use $\hat L_u$ to denote the vector field on $V$
whose value at $x\in V$ is $(x, -ux)\in T_xV$.  Viewed as a differential operator, we have
$$
\hat L_u=-\langle ux\mid {/\hskip -5pt
\partial}\rangle.
$$

The following theorem is extremely useful when we do concrete
computations with Jordan algebras.
\begin{Th}[Jordan Frame]
Let $V$ be a simple euclidean Jordan algebra of rank $\rho$, $x_0\in
V$. Suppose that $x_0$ is non-zero and satisfies equation $x_0^2=\tr x_0\, x_0$. Then there is an
orthogonal basis for $V$: $e_{11}$, \ldots, $e_{\rho\rho}$,
$e_{ij}^\mu$ ($1\le i < j\le \rho$, $1\le\mu \le \delta$) such that

1) each basis vector has length $1\over \sqrt{\rho}$;

2) $e_{ii}^2=e_{ii}$, $(e_{jk}^\mu)^2={1\over 2}(e_{jj}+e_{kk})$;

3) $e_{ii}e_{ij}^\mu=e_{jj}e_{ij}^\mu={1\over 2} e_{ij}^\mu$;

4) $e_{ii}e_{jk}^\mu=0$ if $i\not\in\{j,k\}$, and $e_{ii}e_{jj}=0$ if $i\neq j$;

5) $e_{11}+\cdots+e_{\rho\rho}=e$;

6) $\tr e_{ii}=1$, $\tr e_{ij}^\mu=0$;

7) $x_0$ is a multiple of $e_{11}$.

\end{Th}
In the above theorem, $\{e_{11}, \ldots, e_{\rho\rho}\}$ is called a
{\bf Jordan frame} and the parameter $\delta$ is called the {\bf
degree} of $V$.  If $i<j$, we use $V_{ij}$ to denote
$\mr{span}_{\bb R}\{e_{ij}^\mu\mid 1\le \mu \le \delta\}$ --- the $(i,j)$-{\bf Pierce
component}.

\begin{rmk}
Simple euclidean Jordan algebras are isomorphic if and only if they
have the same rank $\rho$ and same degree $\delta$, as one can verify
from the table below:
\begin{displaymath}
\begin{array}{|c|c|c|c|c|c|}
\hline
V & \Gamma(n) & {\mathcal H}_n(\bb R) & {\mathcal H}_n(\bb C) &{\mathcal H}_n(\bb H) & {\mathcal H}_3(\bb O)\\
\hline
\rho & 2 & n &  n & n &  3\\
\hline
\delta & n-1 & 1 &  2 &  4 &  8\\
\hline
\end{array}
\end{displaymath}
It is also clear from this table and the above theorem that, for the
simple euclidean Jordan algebras, there is one with rank-one,
infinitely many with rank two, four with rank three, and three with
rank four or higher.
\end{rmk}
\section{Tits-Kantor-Koecher Construction}\label{S:TKK}
The Tits-Kantor-Koecher construction yields a simple real Lie
algebras from a simple euclidean Jordan algebra. We begin with the
introduction of the {\bf Jordan triple product}:
$$
\{uvw\} :=u(vw)+w(vu)-(uw)v.
$$ One can check that the Jordan triple product satisfies the following identities:
\begin{eqnarray}\label{JTP}
\{wyz\} & = & \{zyw\},\cr \{uv\{zwy\}\} & = &
\{\{uvz\}wy\}-\{z\{vuw\}y\}+\{zw\{uvy\}\}.
\end{eqnarray}

For a pair $(u,v)\in V\times V$, we introduce the linear map
$S_{uv}$: $V\to V$ by declaring that $$S_{uv}(z):=\{uvz\}.$$ Then
Eq. (\ref{JTP}) is equivalent to the commutation relation
\begin{eqnarray}\label{strR}[S_{uv},S_{zw}] =
S_{\{uvz\}w}-S_{z\{vuw\}}.
\end{eqnarray} Therefore, the span of these
$S_{uv}$, denoted by $\frk{str}(V)$ or simply $\frk{str}$, becomes a real Lie algebra
--- the {\bf structure algebra} of $V$. One can check that
$S_{ue}=S_{eu}=L_u$ and $S_{uv} = [L_u, L_v]+L_{uv}$. Then $S_{uv}'=S_{vu}$, hence
$S_{uv}^*(\langle z\mid)=-\langle S_{vu}(z)\mid$. Therefore, in view of Eq. (\ref{strR}), we see that $z$ and $w$ in $S_{zw}$ transform under $\frk{str}(V)$ as a vector
and a co-vector respectively.

Eq. (\ref{usefulid}) amounts to saying that that $[L_u, L_v]$: $V\to V$ is a
derivation; in fact, any derivation is a linear combination of derivations of this form. The {\bf
derivation algebra} of $V$, denoted by $\frk{der}(V)$ or simply $\frk{der}$, is then a Lie
subalgebra of the structure algebra.

\vskip 10pt The {\bf conformal algebra} of $V$, denoted by
$\frk{co}(V)$ or simply $\frk{co}$, is a further extension of $\frk{str}(V)$. As a vector space, $$\frk{co}(V)=V\oplus \frk{str}(V)\oplus V^*.$$
In the following we shall rewrite $u\in V$ as $X_u$ and $\langle v\mid\, \in V^*$ as $Y_v$.
\begin{Th}[Koecher]\label{koecher} Let $V$ be a simple euclidean Jordan algebra, then
$\frk{co}(V)=V\oplus \frk{str}(V)\oplus V^* $ becomes a simple
real Lie algebra with the defining TKK commutation relations
\begin{eqnarray}\label{CONFA}
\fbox{$\begin{matrix}[X_u, X_v] =0, \quad [Y_u, Y_v]=0, \quad [X_u,
Y_v] = -2S_{uv},\cr\\ [S_{uv},
X_z]=X_{\{uvz\}}, \quad [S_{uv}, Y_z]=-Y_{\{vuz\}},\cr\\
[S_{uv}, S_{zw}] = S_{\{uvz\}w}-S_{z\{vuw\}}
\end{matrix}$}
\end{eqnarray}for $u$, $v$, $z$, $w$ in $V$.
\end{Th}
It is not hard to remember the TKK commutation relations in the
above theorem if one notices that, under
the action of the structure algebra, $u$ and
$v$ in $S_{uv}$ transform as a vector and a co-vector respectively.  It follows from the TKK commutation relations that
\begin{eqnarray}\label{derivation}
[D, L_u]=L_{Du}, \quad [D, X_u]=X_{Du}, \quad [D, Y_u]=Y_{Du}
\end{eqnarray}
for any derivation $D\in\frk{der}$ and any $u\in V$. In particular, we have $[D, L_e]=0$, $[D, X_e]=0$ and $[D, Y_e]=0$.

Let $V$ be a simple euclidean Jordan algebra with an orthonormal
basis $\{e_\alpha\}$ chosen.  Upon recalling the definition of ${/\hskip -5pt
\partial} $, we can introduce differential operators
$$
\hat S_{uv}=-\langle S_{uv}(x)\mid {/\hskip -5pt
\partial} \rangle, \quad \hat X_u=-\langle u\mid {/\hskip -5pt
\partial} \rangle, \quad \hat Y_v=-\langle \{xvx\}\mid {/\hskip -5pt
\partial} \rangle.
$$
One can verify that $\hat S_{ue}=\hat S_{eu}=\hat L_u$ and $\hat S_{uv}=[\hat L_u, \hat L_v]+\hat L_{uv}$.

The following proposition implies that the conformal algebra of $V$ can be
realized as a Lie subalgebra of the Lie algebra of vector fields on
$V$.

\begin{Prop}\label{prop1}
The TKK commutation relations can be realized
by vector fields $\hat S_{uv}$, $\hat X_z$, $\hat Y_w$.

\end{Prop}
\begin{proof}
It is clear that $[\hat X_u, \hat X_v]=0$.

\begin{eqnarray}
[\hat S_{uv}, \hat S_{zw}] &=&[-\langle S_{uv}(x)\mid {/\hskip -5pt
\partial} \rangle,
-\langle S_{zw}(x)\mid {/\hskip -5pt
\partial} \rangle ]\cr
&=&-\langle S_{uv}S_{zw}(x)\mid {/\hskip -5pt
\partial} \rangle + \langle S_{zw}S_{uv}(x)\mid {/\hskip -5pt
\partial} \rangle\cr
&=& -\langle [S_{uv},S_{zw}](x)\mid {/\hskip -5pt
\partial} \rangle\cr
&=&-\langle (S_{\{uvz\}w}-
S_{z\{vuw\}})(x)\mid {/\hskip -5pt
\partial} \rangle\cr
&=&\hat S_{\{uvz\}w} - \hat S_{z\{vuw\}}.\nonumber
 \end{eqnarray}
\begin{eqnarray}
[\hat S_{uv}, \hat X_z] &=&[-\langle S_{uv}(x)\mid
{/\hskip -5pt
\partial}\rangle, -\langle z\mid {/\hskip -5pt
\partial}\rangle
]\cr
&=&- \langle S_{uv}(z)\mid
{/\hskip -5pt
\partial}\rangle=\hat X_{\{uvz\}}.\nonumber
\end{eqnarray}
\begin{eqnarray}
[\hat X_u, \hat Y_v]&=&[-\langle u\mid
{/\hskip -5pt
\partial}\rangle, -\langle \{xvx\}\mid {/\hskip -5pt
\partial}\rangle]\cr
&=&2\langle \{uvx\}\mid
{/\hskip -5pt
\partial}\rangle=-2\hat S_{uv}.\nonumber
\end{eqnarray}
\begin{eqnarray}
[\hat S_{uv}, \hat Y_z] &=&[-\langle S_{uv}(x)\mid
{/\hskip -5pt
\partial}\rangle, -\langle \{xzx\}\mid {/\hskip -5pt
\partial}\rangle
]\cr &=&2\langle \{S_{uv}(x)zx\}\mid
{/\hskip -5pt
\partial}\rangle  -\langle S_{uv}(\{xzx\})\mid {/\hskip -5pt
\partial}\rangle
\cr
&=&2\langle S_{xz}S_{uv}(x)\mid {/\hskip -5pt
\partial}\rangle -\langle S_{uv} S_{xz}(x)\mid {/\hskip -5pt
\partial}\rangle\cr
&=&\langle S_{xz}S_{uv}(x)\mid {/\hskip -5pt
\partial}\rangle-\langle [S_{uv}, S_{xz}](x)\mid {/\hskip -5pt
\partial}\rangle\cr
&=&\langle S_{xz}S_{uv}(x)-S_{\{uvx\}z}(x)\mid {/\hskip -5pt
\partial}\rangle+\langle S_{x\{vuz\}}(x)\mid {/\hskip -5pt
\partial}\rangle\cr
&=& -\hat Y_{\{vuz\}}.\nonumber
\end{eqnarray}

Finally, 
\begin{eqnarray}
[\hat Y_u, \hat Y_v] &=&[-\langle  \{x ux \}\mid {/\hskip -5pt
\partial}\rangle -\langle  \{x vx \}\mid {/\hskip -5pt
\partial}\rangle]\cr
&=& 2\langle \{ \{x ux \}vx\}\mid {/\hskip -5pt
\partial}\rangle - <u\leftrightarrow v>\cr
&=& 2\langle [S_{x v},S_{x u}] x\mid {/\hskip -5pt
\partial}\rangle\\
&=&\langle [S_{x v},S_{x u}] x\mid {/\hskip -5pt
\partial}\rangle - <u\leftrightarrow v>\cr
&=& \langle (S_{\{x vx\}u}-S_{x \{vxu\}})x\mid {/\hskip -5pt
\partial}\rangle - <u\leftrightarrow v>\cr
&=&  \langle S_{\{x vx\}u}x\mid {/\hskip -5pt
\partial}\rangle - <u\leftrightarrow v>\quad \mbox{because  $\{ux v\}=\{vx u\}$}\cr&=&  - \langle [S_{xv}, S_{xu}]x\mid {/\hskip -5pt
\partial}\rangle \\
&=&0,\nonumber
\end{eqnarray}
because it is equal to the negative half of itself. Here,  $<u\leftrightarrow v>$ means the term same as the one on the left except that $u$ and $v$ are interchanged.

\end{proof}

We conclude this section with a few more terminologies and notations. Let $\Omega$ be
the {\bf symmetric cone} of $V$ and $\mr{Str}(V)$ be the {\bf structure group} of $V$.  By definition,
$\Omega$ is the topological interior of $\{x^2\mid x\in V\}$ and 
$$
\mr{Str}(V)=\{g\in GL(V)\mid P(gx) =gP(x)g' \quad \forall x\in V\}.
$$
Here $P(x):=2L_x^2-L_{x^2}$ and it is called the {\bf quadratic representation} of $x$.

We write $V^{\bb C}$ for the complexification of $V$, denote by $T_\Omega$ the tube domain associated with $V$. By definition, $T_\Omega=V\oplus i\Omega$ --- a domain inside $V^{\bb C}$. We say that map $f$: $T_\Omega\to T_\Omega$ is a holomorphic automorphism of $T_\Omega$ if $f$ is invertible and both $f$ and $f^{-1}$ are holomorphic. We use $\mr{Aut}(T_\Omega)$ to denote the group of holomorphic automorphisms
of $T_\Omega$. 

It is a fact that both $\mr{Str}(V)$ and $\mr{Aut}(T_\Omega)$  are Lie groups.  The Lie algebra of $\mr{Str}(V)$ is $\frk{str}(V)$, the Lie algebra of $\mr{Aut}(T_\Omega)$ is $\frk{co}(V)$, and its universal enveloping algebra is called the {\bf TKK algebra} of $V$. The simply connected Lie group with $\frk{co}$ as its Lie algebra, denoted by $\mr{Co}(V)$ or simply $\mr{Co}$, shall be referred to as the {\bf conformal group} of $V$. While the structure algebra is reductive, the conformal algebra is simple.

\section{Cantan Involutions and Vogan Diagrams}\label{S:Vogan}
The complex simple Lie algebras are completely classified by (connected) Dynkin diagrams.
The real simple Lie algebras are completely classified too, and can be represented by Vogan diagrams --- Dynkin diagrams
with some extra information on the Dynkin nodes.

\subsection{Generalities}
Here is a quick review of real simple Lie algebras and Vogan diagrams. Let $\frk g$ be a real simple Lie algebra, and $\langle\,,\,\rangle$ be its Killing form. An involution $\theta$ on $\frk g$ is called a {\bf Cartan involution} if the bilinear form $(X, Y)\mapsto \langle X, \theta(Y)\rangle$ is negative definite. Given a Cartan involution $\theta$, we have the corresponding {\bf Cartan decomposition}:
$$
\frk g=\frk{u}\oplus\frk{p}.
$$
Here, $\frk u$ ($\frk p$ resp.) is the eigenspace of $\theta$ with eigenvalue $1$ ($-1$ resp.). A subalgebra $\frk h$ of $\frk g$ is called a
{\bf $\theta$-stable Cartan subalgebra} of $\frk g$ if $\frk h^{\bb C}$ is a Cartan algebra of $\frk g^{\bb C}$ and $\theta(\frk h)=\frk h$. Here, $\frk h^{\bb C}$ ($\frk g^{\bb C}$ resp.) is the complexification of $\frk h$ ($\frk g$ resp.).

Here are some basic facts:
\begin{itemize}
\item There is a Cartan involution $\theta$ on $\frk g$, unique up
to conjugations.

\item $\mr{span}_{\bb R}\{X+iY\mid X\in \frk u, Y\in \frk p\}$ is a compact Lie algebra.

\item $\theta$-stable Cartan subalgebra of $\frk g$ exists, but are not all conjugate to each other.
\end{itemize}

Given a $\theta$-stable Cartan subalgebra $\frk h$, there is a corresponding root space decomposition:
$$
\frk{g}^{\bb C}=\frk{h}^{\bb
C}\oplus\bigoplus_{\alpha\in\Delta}\frk{g}_\alpha.
$$
A root $\alpha$ w.r.t. $(\frk g^{\bb C}, \frk h^{\bb C})$ is called {\it compact} ({\it non-compact} resp.) if $\frk g_\alpha$ is a subspace of $\frk u^{\bb C}$ ($\frk p^{\bb C}$ resp.). Here is a simple fact on compact or non-compact roots: suppose that root $\alpha$ is either compact or non-compact, then one can choose root vectors $E_\alpha\in{\frk g}_\alpha$,
$E_{-\alpha}\in{\frk g}_{-\alpha}$ with both $\sqrt{-1}(E_\alpha+E_{-\alpha})$ and $E_\alpha-E_{-\alpha}$ in
$\frk{g}$, and an element $H_\alpha$ in $\sqrt{-1}{\frk h}$,  such
that\footnote{In this convention, $H_\alpha\in \sqrt{-1}{\frk h}$ is fixed
by condition $\alpha(H_\alpha)=2$, $E_\alpha$ is unique up to a sign
and $E_{-\alpha}$ is fixed once $E_\alpha$ is fixed.}
\begin{eqnarray}\label{convention}
[H_\alpha, E_{\pm\alpha}] =\pm 2E_{\pm\alpha},\quad [E_\alpha, E_{-\alpha}] = \left\{\begin{matrix} H_\alpha & \mbox{if $\alpha$ is compact}\cr
-H_\alpha & \mbox{if $\alpha$ is non-compact.}\end{matrix}\right.
\end{eqnarray}

We say that a $\theta$-stable Cartan subalgebra of $\frk g$ is {\it maximally compact} if $\dim ({\frk h}\cap \frk u)$ is as large as possible. To get a Vogan diagram, the first step is to find a
{\it maximally compact} $\theta$-stable Cartan subalgebra $\frk{h}$ for
$\frk{g}$. The next step is to chose a simple root system $R$ (or Weyl chamber) for the corresponding root system $\Delta$. Such a $R$ is unique up to the action by the Weyl group $W(\Delta)$. Since $\frk h$ has been chosen to be maximally compact, the roots w.r.t. $(\frk g^{\bb C}, \frk h^{\bb C})$ never vanish on $\frk h\cap \frk u$, hence are either complex-valued or imaginary-valued on $\frk h$. So $R$ splits
into two classes: {\it complex} and {\it imaginary}. Since $R$ is invariant
under complex conjugation, the class of complex simple roots splits further into various conjugate pairs of simple roots. Since the corresponding root space $\frk g_\alpha$ for an imaginary root $\alpha$ is a subspace of $\frk u^{\bb C}$ or $\frk p^{\bb C}$, the class of imaginary simple roots splits further into two subclasses: compact and non-compact.

\begin{Def} By definition, a {\bf Vogan diagram} is a Dynkin diagram with
such an information about its nodes recorded: we paint each
imaginary noncompact node black, connect each conjugate pair of
complex nodes by a two-way arrow, and do nothing to the imaginary
compact nodes.
\end{Def}

Note that one can recover the simple real Lie algebra from one of its Vogan diagrams.  Note also that, in the equal rank case (i.e., the case when $\frk g$ and $\frk u$ have the same rank), $\frk h\subset \frk u$, so every root is either compact or non-compact, and there is no conjugate pair of Dynkin nodes.

\subsection{Analysis for the conformal algebra}
The conformal algebra is a real simple Lie algebra per theorem
\ref{koecher}, so it admits a Cartan involution $\theta$, unique up
to inner automorphisms. Indeed, one can choose $\theta$ such that
$$
\theta(X_u)=Y_u, \quad \theta(Y_u)=X_u, \quad
\theta(S_{uv})=-S_{vu}.
$$  The resulting Cartan decomposition is
$\frk{co}=\frk{u}\oplus\frk{p}$ with
$$\frk{u}=\mr{span}_{\bb R}\{[L_u, L_v], X_w+Y_w\mid u, v, w\in V\}, \quad
\frk{p}=\mr{span}_{\bb R}\{L_u, X_v-Y_v\mid u, v\in V\}.$$
Note that $\frk{u}$ is reductive with center spanned by $X_e+Y_e$ and its semi-simple part  $\bar {\frk{u}}$ is
$$ \mr{span}_{\bb R}\{[L_u, L_v], X_w+Y_w\mid u,
v, w\in ({\bb R}e)^\perp\}.$$
Sometime we need to emphasize the dependence on $V$, then we rewrite $\frk u$ as ${\frk u}(V)$. It is a fact that $\frk{str}$ and $\frk{u}$ are different real forms of the same complex reductive Lie algebra. In fact, one can identify their complexfications as follows:
\begin{eqnarray}\label{id: str=k}
[L_u, L_v] \leftrightarrow [L_u, L_v],\quad
-{i\over 2}(X_w+Y_w)  \leftrightarrow  L_w.
\end{eqnarray}
Now we have another natural chain of real Lie algebras associated
with the Jordan algebra $V$:
\begin{eqnarray}
\frk{der}\subset\bar{\frk{u}}\subset\frk{u}\subset\frk{co}.\nonumber
\end{eqnarray}
We shall use $\tilde{\mr U}$ to denote the closed Lie subgroup of $\mr{Co}$
whose Lie algebra is $\frk u$. Note that  $\tilde{\mr U}$  is a closed maximal
subgroup of $\mr {Co}$ with $\tilde{\mr U}/{\mr Z}$ is compact. (Here $\mr
Z$ is the center of $\mr {Co}$.)

Here is a detailed summary of all real Lie algebras we have
encountered:
\begin{displaymath}
\begin{array}{|c|c|c|c|c|}
\hline
V & \frk{der} & \frk{str} & \frk{u} & \frk{co} \\
\hline
\Gamma(n) & \frk {so}(n) & \frk{so}(n,1)\oplus \bb R & \frk{so}(n+1)\oplus \frk{so}(2) & \frk{so}(n+1,2) \\
\hline
{\mathcal H}_n(\bb R)  & \frk{so}(n) & \frk{sl}(n,{\bb R})\oplus \bb R & \frk{u}(n) &  \frk{sp}(n,{\bb R}) \\
\hline
{\mathcal H}_n(\bb C)  & \frk{su}(n) & \frk{sl}(n,{\bb C})\oplus \bb R & \frk{su}(n)\oplus\frk{su}(n)\oplus \frk{u}(1) &  \frk{su}(n,n)\\
\hline
{\mathcal H}_n(\bb H) & \frk{sp}(n) &  \frk{su}^*(2n)\oplus \bb R & \frk{u}(2n)& \frk{so}^*(4n)  \\
\hline
{\mathcal H}_3(\bb O) & \frk{f}_4 & \frk{e}_{6(-26)}\oplus \bb R & \frk{e}_{6}\oplus \frk{so}(2) & \frk{e}_{7(-25)}  \\
\hline
\end{array}
\end{displaymath}

To get a Vogan diagram for $\frk{co}$, the first step is to find a
maximally compact $\theta$-stable Cartan subalgebra $\frk{h}$ for
$\frk{co}$. Since we are in the equal rank case, $\frk{h}\subset\frk u$, and a root $\alpha$ is either {\em
compact} or {\em non-compact}. Recall that $e_{11}$ denotes the first element of a Jordan frame for $V$.

\begin{Lem}\label{LemmaVogan}
There is a maximally compact $\theta$-stable Cartan subalgebra $\frk h$ for
$\frk{co}$, with respect to which, there is a simple root system
consisting of imaginary roots $\alpha_0$, $\alpha_1$, \ldots,
$\alpha_r$ such that, for $i\ge 1$, $\alpha_i$ is compact with $H_{\alpha_i},
E_{\pm\alpha_i}\in \bar { \frk u}^{\bb C}$,  and $\alpha_0$
is non-compact with
$$
H_{\alpha_0}=i(X_{e_{11}}+Y_{e_{11}}), \quad E_{\pm \alpha_0}={i\over
2}(X_{e_{11}}-Y_{e_{11}})\mp L_{e_{11}}.
$$  

\end{Lem}
\begin{proof}
Let us fix a Jordan frame $\{e_{ii}\mid 1\le i\le \rho\}$. Since $\frk u$ is compact, being an abelian
subalgebra of $\frk u$, ${\frk h}':=\mr{span}_{\bb
R}\{X_{e_{ii}}+Y_{e_{ii}}\mid 1\le i\le \rho\}$ can be extended to a
Cartan subalgebra $\frk{h}$ for $\frk u$, hence a maximally compact
$\theta$-stable Cartan subalgebra $\frk{h}$ for $\frk{co}$. 

Under the action of $\frk h$, we have decomposition $\frk{co}^{\bb C}=\frk u^{\bb C}\oplus \frk p^{\bb C}$. Therefore, if $\alpha$ is a root for $(\frk{co}^{\bb C}, \frk{h}^{\bb C})$, $\frk g_\alpha$ is a subset of either $\frk u^{\bb C}$ or $\frk p^{\bb C}$, so $\alpha$ is either compact or non-compact.

To understand the non-compact roots for $({\frk {co}}^{\bb C}, {\frk h}^{\bb C})$, we introduce
\begin{eqnarray}
E_u^\pm= {\sqrt{-1}\over 2}(X_u-Y_u)\mp L_u,
\quad h_u=\sqrt{-1}(X_u+Y_u), \quad u\in V.\nonumber
\end{eqnarray} 
and verify that
\begin{eqnarray}\label{vogan:key}
\left\{\begin{array}{l}[h_u, E_v^\pm] =  \pm 2 E_{uv}^\pm, \quad
[E_u^+, E_v^-]=-h_{uv}-2[L_u, L_v],\cr
[E_u^+, E_v^+]= [E_u^-,
E_v^-]=0, \quad [h_u, h_v]=4[L_u, L_v].
\end{array}\right.
\end{eqnarray}
The first identity in Eq. (\ref{vogan:key}) implies that, under the action of $\frk h'$, we have the following decomposition of $\frk p^{\bb C}$: 
\begin{eqnarray}
{\frk p}^{\bb C}=\bigoplus_{i\le j}\left({\frk
g}^+_{ij}\oplus {\frk g}^-_{ij}\right)\nonumber
\end{eqnarray}
where $ {\frk g}^\pm_{ii}=\mr{span}_{\bb C}\{E^\pm_{e_{ii}} \}$ and
${\frk g}^\pm_{ij}=\mr{span}_{\bb C}\{E_u^\pm \mid u\in V_{ij}\}$ if
$i< j$. 

Let $\alpha$ be a non-compact root. The decomposition above implies that
$\frk g_\beta\subset \frk g_{ij}^\pm$ for some $i, j$ with $i\le j$. Since $ {\frk g}^\pm_{ii}$ is one-dimensional and Eq. (\ref{vogan:key}) implies that
$$
[E_{e_{ii}}^+, E_{e_{ii}}^-]=-h_{e_{ii}}, \quad [h_{e_{ii}},
E_{e_{ii}}^\pm]=\pm 2E_{e_{ii}}^\pm,
$$
we conclude that, in view of Eq. (\ref{convention}), there is a non-compact root
$\beta_i$ such that ${\frk g}_{\pm \beta_i}={\frk g}_{ii}^\pm$ and
\fbox{$H_{\beta_i}=h_{e_{ii}}$}.
We call $\beta_i$ a non-compact root of {\em type I}.

Suppose that $\alpha$ is a non-compact root of {\em type II}, i.e.,
not of type I. We may assume that ${\frk g}_\alpha\subset {\frk g}_{jk}^+$ for
some $j<k$, then the root vector $E_\alpha\in {\frk g}_\alpha$ is of
the form
$$
E_\alpha=E_z^+
$$ for some $z\in V_{jk}^{\bb C}$. Consequently,
$E_{-\alpha}=E_{\bar z}^-$, and Eq. (\ref{vogan:key}) implies that
\begin{eqnarray}
[E_\alpha, E_{-\alpha}] = -h_{z\bar z}-2[L_z, L_{\bar z}].\nonumber
\end{eqnarray}
So there is a non-compact root $\beta_{jk}^z$ such that
${\frk g}_{\pm \beta_{jk}^z}\subset{\frk g}_{jk}^\pm$ and
\fbox{$H_{\beta_{jk}^z}=h_{z\bar z}+2[L_z, L_{\bar z}]$}. Of course, there are exactly $\delta$-many such $z$. 

In summary, the non-compact roots are $$\pm \beta_i\;  \mbox{and}\;\pm \beta_{jk}^z.$$
and
$$
{\frk g}_{\beta_i}, {\frk g}_{\beta_{jk}^z}\subset {\frk p}_{+}^{\bb C}:=\bigoplus_{i\le j}{\frk
g}^+_{ij}.
$$

To continue the proof, with the help of Eq. (\ref{vogan:key}), we observe that 
$$\beta_{1i}^z(H_{\beta_i})=  \beta_{1i}^z(H_{\beta_1})=1>0,\quad \beta_{1i}^z(H_{\beta_{1i}^z})=2>0, \quad \beta_{1j}^z(H_{\beta_{jk}^w})>0.$$ Here, only the last inequality is not clear, but it can be verified by a
case-by-case study\footnote{it reduces to the
trivial equality $|Im(z_1z_2)|<2$ for $z_1, z_2$ in a division algebra with
$|z_1|^2+|z_2|^2=1$.}. This observation implies that, for $\alpha=\beta_i$ or $\beta_{jk}^z$,  $H_{\beta_1}$ and $H_{-\alpha}$ cannot be in the same Weyl chamber, i.e., $\beta_1$ and $-\alpha$ cannot be in the same simple root system. 

Fixing a simple root system containing the non-compact root $\beta_1$,
all we need to do is to show that this simple root system cannot contain
another non-compact root. Otherwise, this simple system consisted of
simple roots $\gamma_1$, \ldots, $\gamma_t$, $\delta_1$, \ldots,
$\delta_s$ with $\delta_i$ compact, $\gamma_j$ non-compact, $t>1$, $\gamma_1=\beta_1$ and
each $\gamma_j$ is either of the form $\beta_i$ or of the form $\beta_{kl}^z$ due to the conclusion in the previous paragraph, so the first identity in Eq. (\ref{vogan:key}) implies that,
$$
\gamma_j(h_e)=2, \quad 1\le j\le t.
$$

Since the rank of $\frk{co}$ is one
more than the rank of $\bar {\frk u}$, $s$ is less than the rank of
$\bar{\frk u}$, so there is a compact root $\lambda$ such that
$$
\lambda =\sum_i m_i\gamma_i+\sum_jn_j\delta_j
$$ for some non-negative integers $n_j$ and $m_i$ with $\sum_im_i\neq 0$. Since $\alpha(h_e)$ is equal to zero if $\alpha$ is compact, we have
a contradiction:
$$
0=\lambda(h_e)=2\sum_im_i.
$$

\end{proof}
As a corollary of the above lemma, the Vogan diagram we have arrived at for the conformal algebra of a simple euclidean Jordan algebra has no complex nodes and only one node painted black. Here is a pictorial summary of the
Vogan diagrams for the conformal algebras:

\begin{displaymath}
\begin{array}{|c|c|}
\hline
\frk{co} & \mbox{Vogan Diagram} \\
\hline
& \\
& \\
\frk{so}(2,2n) & \setlength{\unitlength}{0.3cm}
\begin{picture}(20,0.25)\linethickness{1mm}
\put(4,0){\circle*{0.3}} \put(6.5,0){\circle{0.3}}
\put(9,0){\circle{0.3}} \put(11.9,0){\circle{0.3}}
\thinlines\put(4,0){\line(1,0){2.3}} \put(7,-0.28){-- -- }
\put(9.2,0){\line(1,0){2.1}}
\put(9,0.2){\line(0,1){2.1}}\put(9.7,0){\line(1,0){2}}
\put(9,2.43){\circle{0.3}}
\end{picture}   \\
& \\
\hline
& \\
\frk{so}(2,2n+1)  &\setlength{\unitlength}{0.3cm}
\begin{picture}(20,0.25)\linethickness{1mm}
\put(4,0){\circle*{0.3}}
 \put(6.5,0){\circle{0.3}} \put(9,0){\circle{0.3}}
\put(10.2,-0.2){$\rangle$}\put(11.5,0){\circle{0.3}}
\thinlines\put(4,0){\line(1,0){2.3}}\put(7,-0.28){-- -- }
\put(9.2,-0.1){\line(1,0){2.1}} \put(9.2,0.1){\line(1,0){2.1}}
\end{picture} \\
& \\
\hline
& \\
\frk{sp}(n, \bb R) & \setlength{\unitlength}{0.3cm}
\begin{picture}(20,0.25)\linethickness{1mm}
\put(4,0){\circle{0.3}}
 \put(6.5,0){\circle{0.3}} \put(9,0){\circle{0.3}}
\put(10,-0.2){$\langle$}\put(11.5,0){\circle*{0.3}}
\thinlines\put(4,0){\line(1,0){2.3}}\put(7,-0.28){-- -- }
\put(9.2,-0.1){\line(1,0){2.1}} \put(9.2,0.1){\line(1,0){2.1}}
\end{picture} \\
&\\
\hline
&\\
\frk{su}(n,n) &  \setlength{\unitlength}{0.3cm}
\begin{picture}(20,0.25)\linethickness{1mm}
\put(4,0){\circle{0.3}}
 \put(6.5,0){\circle{0.3}} \put(9,0){\circle*{0.3}}
\put(11.5,0){\circle{0.3}} \put(14,0){\circle{0.3}}
\thinlines\put(6.7,0){\line(1,0){2.3}} \put(9.2,0){\line(1,0){2.1}}
\put(4.5,-0.28){-- -- } \put(12,-0.28){-- -- }
\end{picture}\\
& \\
\hline
&\\
&\\
\frk{so}^*(4n) & \setlength{\unitlength}{0.3cm}
\begin{picture}(21,0.25)\linethickness{1mm}
\put(4.5,0){\circle{0.3}} \put(7,0){\circle{0.3}}
\put(9.5,0){\circle{0.3}} \put(12,0){\circle*{0.3}} \thinlines
\put(5,-0.28){-- -- } \put(7.2,0){\line(1,0){2.1}}
\put(9.5,0.2){\line(0,1){2.1}}\put(9.7,0){\line(1,0){2.1}}
\put(9.5,2.43){\circle{0.3}}
\end{picture} \\
&\\
\hline
& \\
& \\
\frk{e}_{7(-25)} & \setlength{\unitlength}{0.3cm}
\begin{picture}(20,0.25)\linethickness{1mm}
\put(4,0){\circle{0.3}} \put(6.5,0){\circle{0.3}}
\put(9,0){\circle{0.3}} \put(11.5,0){\circle{0.3}}
\put(14,0){\circle{0.3}}\put(16.5,0){\circle*{0.3}}
\thinlines\put(6.7,0){\line(1,0){2.1}}
\put(9.2,0){\line(1,0){2.1}}\put(11.7,0){\line(1,0){2.1}}\put(4.2,0){\line(1,0){2.1}}\put(14.2,0){\line(1,0){2.1}}
\put(9,0.15){\line(0,1){2.1}}\put(9,2.35){\circle{0.3}}
\end{picture} \\
& \\
\hline
\end{array}
\end{displaymath}

\section{Kepler Cones}\label{S:KM}
The goal in this section is to introduce the Kepler cone, an open
Riemannian manifold which serves as the configuration space for the
J-Kepler problem.

\begin{Definition}[Kepler Cone]\label{D:main1}
Let $V$ be a simple euclidean Jordan algebra. The {\bf Kepler
cone} is a Riemannian manifold whose underlying smooth manifold is
\begin{eqnarray} {\ms
P}:=\left\{x\in V\mid x^2=(\tr x)\, x, \; \tr x >0\right\}
\end{eqnarray}
and its Riemannian metric, referred to as the {\bf Kepler metric},  is the restriction of
\begin{eqnarray}\label{Kepler metric}
ds^2_K:={2\over \rho}ds^2_E-(d\langle e\mid x\rangle)^2
\end{eqnarray}
from $V$ to $\ms P$.
\end{Definition}
We shall also use ${\ms P}$ to denote the Kepler cone. By
introducing coordinates, it is not hard to see that the Kepler cone
is a smooth real affine variety.

${\ms P}$ is called the Kepler cone because it is isometric to the open geometric cone over
{\bf projective space}
\begin{eqnarray}
{\bb P}:=\left\{x\in {\ms P}\mid \tr(x)=\sqrt{{2\rho}}\right\}.
\end{eqnarray}
Here, as Riemannian manifolds, ${\ms P}$ is viewed as $({\ms P}, ds^2_K|_{{\ms P}})$, $\bb R_+\times \bb P$ is viewed as
$\left(\bb R_+\times \bb P, dr^2+r^2\; ds^2_E|_{\bb P}\right)$, and the isometry is
\begin{eqnarray}\iota: \quad {\ms P} &\longrightarrow &  \bb R_+\times \bb P\cr
x & \mapsto & \left({\tr x\over \rho},\; {\sqrt 2}{x\over |x|}\right).
\end{eqnarray}

Note that, being the intersection of ${\ms P }$ with the sphere of radius $\sqrt 2$ and centered at the origin of $V$, ${\bb P}$
is a compact symmetric space of rank-one:
\begin{displaymath}
\begin{array}{|c|c|c|c|c|c|}
\hline
V& \Gamma(n) & {\mathcal H}_n(\bb R) & {\mathcal H}_n(\bb C) &{\mathcal H}_n(\bb H) & {\mathcal H}_3(\bb O)\\
\hline
{\bb P} & {\mr S}^{n-1} & {\bb R}P^{n-1} &  {\bb C}P^{n-1} &  {\bb H}P^{n-1} &  {\bb O}P^2\\
\hline
\end{array}
\end{displaymath}
One can check that the Riemannian metric $ds^2_{\bb P}$ on projective space $\bb P$ is the round metric of the unit spheres for the
Dirac type and is four times the Fubini-Study metric
$$ds^2_{FS}={|dZ|^2\over |Z|^2}-{|Z\cdot d\bar Z|^2\over
|Z|^4}$$ of projective spaces for the other types.

We conclude this section with a technical lemma.

\begin{Lem}\label{KeyLemma}
Let $V$ be a simple euclidean Jordan algebra with rank $\rho$ and degree $\delta$, and $r=\langle e\mid x\rangle$.

i) Let $e_\alpha$ be an orthonormal basis for $V$, then
\begin{eqnarray}
 {\sum_{\alpha,\beta} \mid
[L_{e_\alpha}, L_{e_\beta}]x\rangle\langle [L_{e_\alpha},
L_{e_\beta}]x \mid \over {\rho^2\over 2}\left(1+{\delta\over
4}(\rho-2)\right)}= r\sum_\alpha \mid e_\alpha \rangle \langle
e_\alpha x \mid- \mid x\rangle \langle x\mid
\end{eqnarray}for any $x\in {\ms P}$.

ii) For each $u\in V$ and each $x\in {\ms P}$, the value of $\hat
L_u$ at $x$ is a tangent vector of ${\ms P}$ at $x$. So $\hat L_u$
descends to a differential operator on the Kepler cone.

iii) Let $\lambda_u= {(\rho/2-1)\delta\over 2}{\langle u\mid
x\rangle\over r}+{\rho \delta\over 4}\langle u\mid e\rangle$,
$\mr{vol}_{{\ms P}}$ be the volume element on ${\ms P}$, and ${\ms
L}_u$ be the Lie derivative with respect to vector field $\hat L_u$
on the Kepler cone. Then
\begin{eqnarray}\label{volume}
{\ms L}_u\left({1\over r}\mr{vol}_{{\ms P}}\right)=-2\lambda_u
{1\over r} \mr{vol}_{{\ms P}}.
\end{eqnarray}
Consequently, $\tilde L_u:=\hat L_u-\lambda_u$ is a skew-hermitian
operator with respect to inner product
$$
(\psi_1, \psi_2):=\displaystyle\int_{\ms
P}\overline{\psi_1}\,\psi_2\,{1\over r}\mr{vol}_{{\ms P}}
$$ for compactly-supported smooth functions on ${\ms P}$.

\end{Lem}
\begin{proof} i) Since both sides of the identity are homogeneously quadratic in $x$, one
may assume that $\tr x=1$. Choosing a Jordan frame $\{e_{11}$, \ldots, $e_{\rho\rho}\}$ with
$e_{11}=x$ and an associated orthonormal basis for $V$, the detailed proof then becomes just a straightforward
computation, so we skip it.

ii) Let $x_0\in \ms P$. Since $x_0^2=\tr x_0\, x_0$ and $\tr x_0>0$,  one can write
$x_0=\tr x_0\, e_{11}$ so that $e_{11}^2=e_{11}$ and $\tr e_{11}=1$. Extending $e_{11}$ to
a Jordan frame $\{e_{11}, \ldots, e_{\rho\rho}\}$ for $V$, then we can
decompose $V$ orthogonally into the direct sum of the Pierce
components $V_{ij}$ ($i\le j$). Then
$$
u x_0\in\bigoplus_{j=1}^\rho V_{1j}.
$$ 
By linearizing equation $x^2=\tr x\, x$ at $x_0$, it is clear that the tangent space of $\ms P$ at $x_0$, when translated to $0$, is exactly
$\bigoplus_{j=1}^\rho V_{1j}$. Therefore,
$$
\hat L_u|_{x_0}=(x_0, -ux_0)\in T_{x_0}\ms P.
$$

(In fact, one can show that the structure group of $V$ acts on $\ms P$ transitively.)

iii) We wish to prove identity (\ref{volume}) at $x_0\in {\ms P}$.
To do that we need to choose a local coordinate system for ${\ms P}$
around $x_0$ and do the computations. We may choose a Jordan frame $\{e_{11}$, \ldots, $e_{\rho\rho}\}$
such that $x_0=ae_{11}$ for some $a>0$. Write
$$x=\sum_{i=1}^\rho x_{ii}e_{ii}+\sum_{1\le i<j\le \rho}^{1\le \mu \le \delta} x_{ij}^\mu e_{ij}^\mu,$$
by solving equation $x^2=\tr x\, x$, we know that $x_{11}$, $x_{1j}^\alpha$'s are independent
real variables and the Taylor expansion of the other variables
starts at quadratic terms in $(x_{11}-a)$, $x_{1j}^\alpha$'s.
Therefore,
$$
ds^2_K={1\over \rho^3}\left[(dx_{11})^2+2\sum_{2\le j\le \rho}^{1\le
\alpha\le \delta} (dx_{1j}^\alpha)^2\right]+O(|x-x_0|^2)
$$
and
$$
\mr{vol}_{{\ms P}}=b\,
dx_{11}\wedge(\wedge_{j=2}^\rho(\wedge_{\alpha=1}^\delta
dx_{1j}^\alpha))+O(|x-x_0|^2)
$$ with $b$ being a constant. Since
\begin{eqnarray}
{\ms L}_u(dx_{11}) &= & d {\ms L}_u (x_{11})=-d\langle ux\mid \rho e_{11}\rangle= -d\langle x\mid
\rho e_{11}u\rangle\cr &=& -\langle e_{11}\mid \rho e_{11}u\rangle
dx_{11}-\sum_{2\le i\le \rho}^{1\le \alpha\le \delta}\langle
e_{1i}^\alpha\mid \rho e_{11}u\rangle dx_{1i}^\alpha\;,\cr {\ms
L}_u(dx_{1j}^\beta) &=& -\langle e_{1j}^\beta\mid \rho
e_{1j}^\beta u\rangle dx_{1j}^\beta-\langle  e_{11}\mid \rho
e_{1j}^\beta u\rangle dx_{11}\cr &&-\sum_{(i, \alpha)\neq (j,
\beta)}\langle e_{1i}^\alpha\mid \rho e_{1j}^\beta u\rangle
dx_{1i}^\alpha\;,\nonumber
\end{eqnarray}
we have
\begin{eqnarray}
{\ms L}_u(\mr{vol}_{{\ms P}})|_{x_0} &= &\left.-\left(\langle
e_{11}\mid \rho e_{11}u\rangle+\sum_{j\ge 2, 1\le\beta\le \delta}\langle
e_{1j}^\beta\mid \rho e_{1j}^\beta u\rangle\right)\mr{vol}_{\ms
P}\right|_{x_0}\cr &= &\left.-\rho\left(\langle e_{11}\mid
u\rangle+{\delta\over 2}\sum_{j\ge 2}\langle  e_{11}+e_{jj}\mid
u\rangle\right)\mr{vol}_{{\ms P}}\right|_{x_0}\cr &=
&\left.-\rho\left((1+(\rho/2-1)\delta)\langle e_{11}\mid u\rangle+{\delta\over
2}\langle  e\mid u\rangle\right)\mr{vol}_{\ms
P}\right|_{x_0}.\nonumber
\end{eqnarray}
On the other hand,
$$
\left.{\ms L}_u({1\over r})\right|_{x_0}=\left.{\langle u\mid
x\rangle\over r}\cdot {1\over r}\right|_{x_0}=\left.\rho\langle
u\mid e_{11}\rangle \cdot {1\over r}\right|_{x_0}.
$$
Therefore,
\begin{eqnarray}
{\ms L}_u({1\over r}\mr{vol}_{{\ms P}})|_{x_0} &= &
\left.-\rho\left((\rho/2-1)\delta\langle e_{11}\mid u\rangle+{\delta\over
2}\langle  e\mid u\rangle\right){1\over r}\mr{vol}_{\ms
P}\right|_{x_0}\cr &=& -2\lambda_u {1\over r}\mr{vol}_{\ms
P}|_{x_0}.\nonumber
\end{eqnarray}
Then
\begin{eqnarray}
(\tilde L_u\psi_1, \psi_2)+(\psi_1, \tilde L_u\psi_2) & = &
\displaystyle\int_{{\ms P}}{\ms L}_u(\overline{\psi_1}\,\psi_2\,
{1\over r}\mr{vol}_{{\ms P}})\cr &=& \displaystyle\int_{{\ms P}}d
\iota_{\hat L_u}(\overline{\psi_1}\,\psi_2\, {1\over
r}\mr{vol}_{{\ms P}}) =0.\nonumber
\end{eqnarray}
Here $\iota_{\hat L_u}$ is interior product of differential form
with vector field $\hat L_u$.

\end{proof}

\section{The Hidden Action on the
Kepler Cones}\label{S:DS}

Our recent investigation of the Kepler-type problems leads to the discovery
of the hidden action of the conformal algebra on the Kepler cone. By
turning arguments backward, we can say that it is this hidden action
that is responsible for the existence of Kepler-type problems.

We begin with some generalities. For smooth manifold $M$, we use
$\frk X(M)$ to denote the Lie algebra of (smooth) vector fields on
$M$ and ${\ms D}(M)$ to denote the algebra of smooth
(real) differential operators on $M$.

Let $A$ be an associative algebra with identity over $\bb R$. We say
that {\bf $A$ acts on $M$ hiddenly} if there is an algebra
homomorphism from $A$ into ${\ms D}(M)\otimes_{\bb R} \bb C$. For
example, if $A=\bb R[t]$ (the polynomial algebra over $\bb R$ in
single variable $t$), then the algebra homomorphism $A\to {\ms
D}(M)\otimes_{\bb R}\bb C$ sending $t$ to the Laplace operator on
$M$, defines a hidden action of $A$ on $M$.

Let $\frk g$ be a real Lie algebra. We say that {\bf $\frk g$ acts
on $M$} if there is a Lie algebra homomorphism from $\frk g$ into
$\frk X(M)$; and we say that {\bf $\frk g$ acts on $M$ hiddenly} if
the universal enveloping algebra of $\frk g$ acts on $M$ hiddenly.
It is clear that, if $\frk g$ acts on $M$, then it acts on $M$
hiddenly; however, the converse may not be true. Note that, $\frk X(M)$ acts on $M$, but
$\ms D(M)$ acts on $M$ only hiddenly.

\vskip 10pt Let $V$ be a simple euclidean Jordan algebra. Since the
automorphisms of $V$ leave the Kepler cone invariant, the derivation
algebra, being the Lie algebra of automorphism group, acts on the
Kepler cone. The recent investigation of the Kepler-type problems leads to the following fact: {\it there is a natural
hidden action of the conformal algebra on the Kepler cone which
extends the action of the derivation algebra}.

To introduce the hidden action, we fix an orthonormal basis
$e_\alpha$ for $V$ and recall that $\tilde L_u=\hat L_u-\lambda_u$
where
$$\lambda_u= {(\rho/2-1)\delta\over 2}{\langle u\mid x\rangle\over
\langle e\mid x\rangle}+{\rho \delta\over 4}\langle u\mid e\rangle.$$ For
$u, v\in V$, we introduce differential operators
\begin{eqnarray}
\fbox{ ${\tilde S}_{uv}:=[\tilde L_u, \tilde L_v]+\tilde L_{uv},
\quad{\tilde X}_u:=-i[{\tilde L}_u, X], \quad {\tilde
Y}_v:=-i\langle v\mid x\rangle $}
\end{eqnarray} where
\begin{eqnarray}\label{defX}
X=-{1\over \langle e\mid x\rangle}\left({\hat
L}_e^2-\left((\rho-1)\delta-1\right){\hat L}_e+A\sum_{\alpha,\beta}[{\hat L}_{e_\alpha}, {\hat L}_{e_\beta}]^2+B\right)
\end{eqnarray}
with $A$ and $B$ being constants depending only on the Jordan
algebra. Note that $X$ is independent of the choice of the
orthonormal basis $e_\alpha$ and
$$
\tilde S_{ue}=\tilde S_{eu}=\tilde L_u.
$$

In view of part ii) of Lemma \ref{KeyLemma}, ${\tilde S}_{uv}$, $\tilde X_u$ and
$\tilde Y_v$ all descend to differential operators on the Kepler
cone.

\begin{Lem}\label{lemma}

i) There is a unique constant $A$ in Eq. (\ref{defX}), such that, as
differential operator on ${\ms P}$,
\begin{eqnarray}
\fbox{$[X, \langle u\mid x\rangle]=2\tilde L_u$}
\end{eqnarray} for any $u\in V$.  In fact
\begin{eqnarray}\label{Avalue}
A^{-1}={\rho^2\over 2}\left(1+{\delta\over 4}(\rho-2)\right).
\end{eqnarray}

ii) Let $\Delta_{\bb P}$ be the Laplace operator on $\bb P$ and
$A$ be the number in Eq. (\ref{Avalue}). Then
\begin{eqnarray}
\Delta_{\bb P}=A\sum_{\alpha,\beta}[{\hat L}_{e_\alpha}, {\hat
L}_{e_\beta}]^2
\end{eqnarray} as differential operators on on $\bb
P$. Consequently,
\begin{eqnarray}\label{Xdef}
\fbox{$X=-\langle e\mid x\rangle\Delta_{{\ms P}}-{B\over \langle
e\mid x\rangle}$.}
\end{eqnarray}
Here, $\Delta_{{\ms P}}$ is the Laplace operator on $\ms P$.
\end{Lem}
\begin{proof}
i) For simplicity, we write $\langle x\mid e\rangle$ as $r$. Since
\begin{eqnarray}
[X, \langle u\mid x\rangle ]&=& -{1\over r}\left[{\hat
L}_e^2-((\rho-1)\delta-1)\hat L_e+A\sum_{\alpha,\beta}[\hat L_{e_\alpha}, \hat L_{e_\beta}]^2, \langle u\mid x\rangle\right]\cr &=& -{1\over
r}\left(-2\langle u\mid x\rangle\tilde L_e+[A\sum_{\alpha,\beta}[\hat L_{e_\alpha}, \hat L_{e_\beta}]^2, \langle u\mid x\rangle]\right),\nonumber
\end{eqnarray}we just need to show that
\begin{eqnarray}\label{identity}
\fbox{$[A\sum_{\alpha,\beta}[\hat L_{e_\alpha}, \hat L_{e_\beta}]^2, \langle u\mid
x\rangle]=-2r\tilde L_u+2\langle u\mid x\rangle\tilde L_e$}
\end{eqnarray} or
\begin{eqnarray}
\left\{\begin{array} {rcl} A\sum_{\alpha,\beta}\langle
[L_{e_\alpha}, L_{e_\beta}]u\mid x\rangle\, [\hat L_{e_\alpha}, \hat L_{e_\beta}] & = & -r\hat L_u+\langle u\mid x\rangle \hat L_e\cr\\
A\langle \sum_{\alpha,\beta}[L_{e_\alpha}, L_{e_\beta}]^2 u\mid
x\rangle & = & -{\rho \delta\over 2}(\langle u\mid x\rangle-r \langle u\mid e\rangle)
\end{array}\right.
\end{eqnarray}
for any $u\in V$ and any $x\in {\ms P}$. So it suffices to show that
\begin{eqnarray}\label{trivialid'}
 A\sum_{\alpha,\beta} \mid
[L_{e_\alpha}, L_{e_\beta}]x\rangle\langle [L_{e_\alpha},
L_{e_\beta}]x \mid = r\sum_\alpha \mid e_\alpha \rangle \langle
e_\alpha x \mid- \mid x\rangle \langle x\mid
\end{eqnarray}
and
\begin{eqnarray}\label{trivialid}
A\sum_{\alpha,\beta}[L_{e_\alpha}, L_{e_\beta}]^2 x = -{\rho \delta\over
2} (x-re)
\end{eqnarray} for any $x\in {\ms P}$.  Eq. (\ref{trivialid'}) is the content of part i) of Lemma \ref{KeyLemma} with
$$
A^{-1}={\rho^2\over 2}\left(1+{\delta\over 4}(\rho-2)\right).
$$
Eq. (\ref{trivialid}) is clear except that we don't know what the
constant $A$ is. Here is a way to find $A$: taking inner product
with $x$, we have
\begin{eqnarray}\label{trivialid''}
A\sum_{\alpha,\beta}||[L_{e_\alpha}, L_{e_\beta}] x||^2 =
{(\rho-1)\delta\over 2}||x||^2.
\end{eqnarray}On the other hand, by taking the trace of Eq. (\ref{trivialid'}), we also arrive at Eq. (\ref{trivialid''}); so
Eq. (\ref{trivialid}) is a consequence of Eq. (\ref{trivialid'}).

ii) In view of identity (\ref{identity}) and the fact that both
$\Delta_{\bb P}$ and $A\sum_{\alpha,\beta}[{\hat L}_{e_\alpha},
{\hat L}_{e_\beta}]^2$ are 2nd order differential operators without
the constant terms, it suffices to show that
\begin{eqnarray}\label{identity'}
[\Delta_{\bb P}, \langle u\mid x\rangle]=-2r\tilde L_u+2\langle
u\mid x\rangle\tilde L_e.
\end{eqnarray}
To verify identity (\ref{identity'}) at a point $x_0\in \bb P$, we choose
a Jordan frame $\{e_{11}, \ldots, e_{\rho\rho}\}$ such that $x_0=\sqrt{2\rho}\,e_{11}$ and write
variable $$x=\sum_{1\le i\le \rho} x_{ii}e_{ii}+\sum_{1\le i<j\le
\rho}^{1\le \alpha\le \delta}x_{ij}^\alpha e_{ij}^\alpha$$ where
$e_{ii}$'s, $e_{ij}^\alpha$'s are mutually orthogonal with length
$1\over \sqrt{\rho}$.

By solving equation $x^2=\tr x\, x$ and $\tr x =\sqrt{2\rho}$, we
know that $x_{1j}^\alpha$'s ($j>1$) are independent real variables and the
Taylor expansion of other variables has no linear terms in
$x_{1j}^\alpha$'s. Therefore, around point $x_0$, we have
\begin{eqnarray}
ds^2_{\bb P} & = & ds^2_E|_{\bb P}={1\over \rho}\sum_{2\le i\le
\rho}^{1\le \alpha \le \delta}(dx_{1i}^\alpha)^2+O(|x-x_0|^2),\cr \langle
u\mid x\rangle &=& \langle u\mid x_0\rangle+\sum_{2\le i\le \rho}^{1\le \alpha \le \delta}\langle
u\mid x_{1i}^\alpha e_{1i}^\alpha \rangle+O(|x-x_0|^2).\nonumber
\end{eqnarray}
Thus
\begin{eqnarray}
\mbox{LHS of Eq. (\ref{identity'})}|_{x_0}  & = & \left.\rho\sum_{2\le
i\le \rho}^{1\le \alpha \le \delta}[{\partial^2\over
\partial (x_{1i}^\alpha)^2},\langle u\mid x\rangle]\right|_{x_0}\\
&=& \left. 2\rho\sum_{2\le i\le \rho}^{1\le \alpha \le \delta}\langle
u\mid e_{1i}^\alpha\rangle{\partial\over
\partial x_{1i}^\alpha}\right |_{x_0}\cr
&& +\left.\rho\sum_{2\le i\le \rho}^{1\le \alpha \le \delta}
{\partial^2\over
\partial (x_{1i}^\alpha)^2}(\langle u\mid x\rangle)\right|_{x_0}.\nonumber
\end{eqnarray}
On the other hand,
\begin{eqnarray}
\mbox{RHS of Eq. (\ref{identity'})}|_{x_0} & = &\left.
2\rho\sum_{2\le i\le \rho}^{1\le \alpha \le \delta}\langle u\mid
e_{1i}^\alpha\rangle{\partial\over
\partial x_{1i}^\alpha}\right |_{x_0}\\
&& +\delta\sqrt{\rho/2}(-\rho\langle u\mid e_{11}\rangle +\langle u\mid
e\rangle).\nonumber
\end{eqnarray}
Therefore, all need to do is to verify that
\begin{eqnarray}
\left.\sum_{2\le i\le \rho}^{1\le \alpha \le \delta} {\partial^2\over
\partial (x_{1i}^\alpha)^2}(\langle u\mid x\rangle)\right|_{x_0}={\delta\over \sqrt{2\rho}}(-\rho\langle u\mid e_{11}\rangle +\langle u\mid
e\rangle)\nonumber
\end{eqnarray}
or
\begin{eqnarray}\label{finalid}
\left.\sum_{2\le i\le \rho}^{1\le \alpha \le \delta} {\partial^2\over
\partial (y_{1i}^\alpha)^2}(\langle u\mid y\rangle)\right|_{y=0}=\delta(-\rho\langle u\mid e_{11}\rangle +\langle u\mid
e\rangle),
\end{eqnarray}
here $y$ satisfies the following condition:
\begin{eqnarray}
2e_{11}y+y^2=y, \quad \tr y=0,
\end{eqnarray}
so
$$
y=t+\rho\left[-\langle t^2\mid e_{11}\rangle e_{11}+\sum_{j=2}^\rho
\langle t^2\mid e_{jj}\rangle e_{jj}+ \sum_{2\le k<l\le \rho}^{1\le
\beta \le \delta} \langle t^2\mid e_{kl}^\beta\rangle
e_{kl}^\beta\right]+o(t^2),
$$
where $t$ is a free parameter taking values in $\oplus_{j=2}^\rho V_{1j}$.

So, in view of the fact that $(e_{1i}^\alpha)^2={1\over
2}(e_{11}+e_{ii})$, we have
\begin{eqnarray}
\mbox{LHS of Eq. (\ref{finalid})}&=& \left.\sum_{2\le i\le \rho}^{1\le \alpha \le \delta} {\partial^2\over
\partial (t_{1i}^\alpha)^2}(\langle u\mid y\rangle)\right|_{t=0}\cr
&=& 2\rho\sum_{2\le i\le
\rho}^{1\le \alpha\le \delta}[-\langle (e_{1i}^\alpha)^2\mid
e_{11}\rangle \langle u\mid e_{11}\rangle \cr &&\quad\quad\quad\quad +\sum_{j=2}^\rho
\langle (e_{1i}^\alpha)^2\mid e_{jj}\rangle \langle u\mid
e_{jj}\rangle\cr &&\quad\quad\quad\quad + \sum_{2\le k<l\le \rho}^{1\le \beta \le \delta}
\langle (e_{1i}^\alpha)^2\mid e_{kl}^\beta\rangle \langle u\mid
e_{kl}^\beta\rangle ]\cr
&=& \sum_{2\le i\le \rho}^{1\le \alpha\le
\delta}(- \langle u\mid e_{11}\rangle+\langle u\mid e_{ii}\rangle )\cr
&=& \delta\sum_{2\le i\le \rho}(- \langle u\mid e_{11}\rangle+ \langle
u\mid e_{ii}\rangle )\cr &=& \delta(- (\rho-1)\langle u\mid
e_{11}\rangle+ \langle u\mid e-e_{11}\rangle )\cr &=&\delta(- \rho\langle
u\mid e_{11}\rangle+ \langle u\mid e\rangle )\cr &=& \mbox{RHS of
Eq. (\ref{finalid})}.\nonumber
\end{eqnarray}

\end{proof}
\subsection{Hidden Action}The following theorem implies that there is a hidden action of the conformal algebra
on the Kepler cone.
\begin{Thm}[Hidden Action/Dynamical Symmetry]\label{Main1}
Let $V$ be a simple euclidean Jordan algebra with rank $\rho\ge 2$
and degree $\delta$.  There are unique constants $A$ and $B$ in Eq. (\ref{defX}) such that, as differential operators on the Kepler
cone, ${\tilde S}_{uv}$, ${\tilde  X}_u$ and ${\tilde Y}_v$ satisfy
the TKK commutation relation (\ref{CONFA}) for the conformal algebra,
i.e.,
\begin{eqnarray}
\begin{matrix}[\tilde X_u, \tilde X_v] =0, \quad [\tilde Y_u, \tilde Y_v]=0, \quad [\tilde X_u,
\tilde Y_v] = -2\tilde S_{uv},\cr\\ [\tilde S_{uv},
\tilde X_z]=\tilde X_{\{uvz\}}, \quad [\tilde S_{uv}, \tilde Y_z]=-\tilde Y_{\{vuz\}},\cr\\
[\tilde S_{uv}, \tilde S_{zw}] = \tilde S_{\{uvz\}w}-\tilde
S_{z\{vuw\}}
\end{matrix}\nonumber
\end{eqnarray}for $u$, $v$, $z$, $w$ in $V$. In fact,
$$
A={2/\rho^2\over 1+{\delta\over 4}(\rho-2)},\quad B={\delta\over
8}(\rho-2)\left(\left({3\rho\over 2}-1\right)\delta-2\right).
$$\end{Thm}
\begin{rmk}Here are the more explicit values of $A$ and $B$:
\begin{displaymath}
\begin{array}{|c|c|c|c|c|c|} \hline
V & \Gamma(n) & {\mathcal H}_n(\bb R) & {\mathcal H}_n(\bb C) &{\mathcal H}_n(\bb H) & {\mathcal H}_3(\bb O)\\
\hline
&&&&&\\
A & {1\over 2} & {8\over n^2(n+2)} & {4\over n^3} & {2\over n^2(n-1)} & {2\over 27}\\
&&&&&\\
B & 0 & {3(n-2)^2\over 16} & {(n-2)(3n-4)\over 4} & 3(n-1)(n-2) & 26\\
&&&&&\\
\hline
\end{array}
\end{displaymath}
In the case  $\rho=1$, a similar theorem is valid except
that $A=0$ and $B$ is not unique.
\end{rmk}

\begin{proof}
For any function $f$ on $V$ and $u, v\in V$, we have $$(\tilde
Y_u\tilde Y_v) (f)(x)=-\langle u\mid x\rangle\langle v\mid x\rangle f(x),$$
so $[\tilde Y_u, \tilde Y_v]=0$ for any $u, v\in V$. The rest of the
proof is divided into four steps.

\vskip 10pt \underline{Step one}: Verify that $[\tilde S_{uv},
\tilde Y_z]=-\tilde Y_{\{vuz\}}$.

This is a simple computation:
\begin{eqnarray}
[\tilde S_{uv}, \tilde Y_z] & = & [\hat S_{uv}, -i\langle z\mid
x\rangle]= i\langle z\mid S_{uv}(x)\rangle\cr &= &i\langle S_{vu}(z)\mid
x\rangle=i\langle \{vuz\}\mid x\rangle \cr &=&-\tilde
Y_{\{vuz\}}.\nonumber
\end{eqnarray}

\underline{Step two}: Verify that
\begin{eqnarray}\label{step2}
[\tilde S_{uv}, \tilde
S_{zw}]=\tilde S_{\{uvz\}w}-\tilde S_{z\{vuw\}}.
\end{eqnarray}

It is easy to see that $\tilde S_{uv} = \hat S_{uv}-\lambda_{uv}$. Thanks to Proposition \ref{prop1}, $[\hat S_{uv}, \hat
S_{zw}]=\hat S_{\{uvz\}w}-\hat S_{z\{vuw\}}$, so all we need to
check is that
$$
\lambda_{\{uvz\}w}-\lambda_{z\{vuw\}}=[\hat S_{uv},
\lambda_{zw}]-[\hat S_{zw}, \lambda_{uv}],
$$ i.e.,
$$
\lambda_{\{uvz\}w}-\lambda_{z\{vuw\}}=-
\lambda_{\{vu(zw)\}}+\lambda_{\{wz(uv)\}}. 
$$ 
Since $\lambda_u$ is linear in $u$, the last equality is implied by identity
$$
L_{\{uvz\}}-L_z S_{vu}=-S_{vu}L_z+S_{(uv)z},\; \mbox{i.e.},\; L_{\{uvz\}}=S_{(zv)u}-S_{v(zu)}+S_{(uv)z}
$$
or equivalently
$$
0=[L_{zv}, L_u]-[L_v, L_{zu}]+[L_{uv}, L_z]
$$ --- the polarization of identity $[L_{u^2}, L_u]=0$.

\underline{Step three}. Verify that $[\tilde X_u, \tilde
Y_v]=-2\tilde S_{uv}$, i.e.,
\begin{eqnarray}\label{step3}
\fbox{$[[\tilde L_u, X], \langle v\mid x\rangle]=2\tilde S_{uv}$.}
\end{eqnarray}

This is a consequence of part i) of Lemma \ref{lemma}:
\begin{eqnarray}
[[\tilde L_{u}, X],\langle x\mid v\rangle]&=&[[\langle x\mid
v\rangle, X],\tilde L_{u}]+[[\tilde L_{u},\langle x\mid v\rangle],
X]\cr &=& [-2\tilde L_v, \tilde L_u]+[-\langle x\mid uv\rangle,
X]\cr &=& [-2\tilde L_v, \tilde L_u]+2\tilde L_{uv}=2\tilde
S_{uv}.\nonumber
\end{eqnarray}

\underline{Step four}. Verify that $[\tilde S_{uv}, \tilde
X_z]=\tilde X_{\{uvz\}}$, i.e., $[\tilde S_{uv}, [\tilde L_z,
X]]=[\tilde L_{\{uvz\}}, X]$.

Note that $X$ is invariant under $\frk{der}$ and $\tilde
S_{eu}=\tilde S_{ue}=\tilde L_u$, in view of Eq. (\ref{step2}), we have
\begin{eqnarray}
[\tilde S_{uv}, [\tilde L_z, X]] &=& [X, [\tilde L_z,\tilde S_{uv}]]+[\tilde L_z, [\tilde S_{uv}, X]]\cr &=&[\tilde L_{\{uvz\}}-\tilde
S_{z(vu)}, X]+[\tilde L_z, [\tilde S_{uv}, X]]\cr &=&[\tilde
L_{\{uvz\}}-\tilde L_{z(vu)}, X]+[\tilde L_z, [\tilde L_{uv}, X]]\cr
&=&[\tilde L_{\{uvz\}}, X]-[\tilde L_{z(uv)}, X]+[\tilde L_z,
[\tilde L_{uv}, X]].\nonumber
\end{eqnarray}
So it suffices to show that
\begin{eqnarray}\label{Lid}
\fbox{$[\tilde L_{uv}, X]=[\tilde L_u, [\tilde L_v, X]]$}
\end{eqnarray} for any $u, v\in V$. To prove it, we let
$$O=[\tilde L_{uv}, X]-[\tilde L_u, [\tilde L_v, X]],$$
and show that, 1) for any $z\in V$, $O_z:=[O, \langle x\mid
z\rangle]=0$, 2) $O(1)=0$. Eq. (\ref{step3}) implies that
\begin{eqnarray}
[[\tilde L_u, [\tilde L_v, X]], \langle x\mid z\rangle] & = &
[[\langle x\mid z\rangle, [\tilde L_v, X]], \tilde L_u]+[[\tilde
L_u, \langle x\mid z\rangle], [\tilde L_v, X]]\cr
&=& [[\langle x\mid z\rangle, [\tilde L_v, X]], \tilde L_u]+[-\langle x\mid uz\rangle, [\tilde L_v, X]]\cr
&=& [-2\tilde S_{vz},\tilde L_u]+2\tilde S_{v(uz)}.\nonumber
\end{eqnarray}
So, in view of Eq. (\ref{step2}), we have
\begin{eqnarray}
O_z/2 & = & \tilde S_{(uv)z}+ [\tilde S_{vz},\tilde L_u]-\tilde
S_{v(uz)}\cr 
&=&  \tilde S_{(uv)z}-\tilde
S_{v(uz)} - [\tilde S_{ue}, \tilde S_{vz}]\cr
&=&0.\cr 
\end{eqnarray}
The proof of $O(1)=0$ is a long computation, so its details are
provided in the appendix.

\underline{Step five}. Verify that $[\tilde X_u, \tilde X_v]=0$,
i.e., $[[\tilde L_u, X], [\tilde L_v, X]]=0$.

Eq. (\ref{Lid}) implies that
\begin{eqnarray}
[[\tilde L_u, X], [\tilde L_v, X]] &=& [[\tilde L_u,[\tilde L_v,
X]],X]+[[\tilde L_v, X],X], \tilde L_u] \cr &=& [[\tilde L_{uv},
X],X]+[[[\tilde L_v, X],X], \tilde L_u].\nonumber
\end{eqnarray}
So it suffices to show that
\begin{eqnarray}\label{Lid2}
\fbox{$[[\tilde L_u, X], X]=0$}
\end{eqnarray} for any $u\in V$. To prove Eq. (\ref{Lid2}), as in step four, we let $\ms O=[[\tilde
L_u, X], X]$, and show that, 1) for any $z\in V$, $\ms O_z:=[\ms O,
\langle x\mid z\rangle]=0$,  2) $\ms O(1)=0$.

In view of Eqs. (\ref{Lid}) and (\ref{step3}), part i) of Lemma \ref{lemma} and the
fact that $X$ is invariant under the action of $\frk {der}$, we have
\begin{eqnarray}
[\ms O, \langle x\mid z\rangle]&=& [[\langle x\mid z\rangle,X],
[\tilde L_u, X]]+[[[\tilde L_u,X], \langle x\mid z\rangle], X]\cr
&=& -2[\tilde L_z, [\tilde L_u, X]]+[2\tilde S_{uz}, X]\cr
&=& -2[\tilde L_z, [\tilde L_u, X]]+[2\tilde L_{uz}, X]=0.\cr
\end{eqnarray}
The proof of $\ms O(1)=0$ is a long computation, so its details are
provided in the appendix.

\end{proof}
\subsection{Quadratic Relations for the Hidden Action}
Let $V_0$ be the orthogonal complement of $e$ in $V$, $D=\dim V_0$ and $\{e_\alpha\}_{1\le \alpha\le D}$ be the orthonormal basis for $V_0$. We write $e$ as $e_0$ and $\langle x\mid e_\alpha\rangle$ as $x_\alpha$. Note that
$\tilde S_{eu}=\tilde S_{ue}=\tilde L_u$.
\begin{Thm}(The Primary Quadratic Relation)
\begin{eqnarray}\label{case1}
{2\over \rho} \sum_{0\le \alpha\le D} \tilde L_{e_\alpha}^2-\tilde L_e^2-{1\over 2}\{\tilde X_e,\tilde Y_e\}= -a.
\end{eqnarray}
 Here, $a={\rho \delta\over 4}(1+{(\rho-2)\delta\over 4})$.
 \end{Thm}
 \begin{proof}  
Let $\ms O$ be the left hand side of Eq. (\ref{case1}). We just need to show that, on the Kepler cone we have 1) $[\ms O, \tilde Y_u]=0$ and 2) $\ms O(1)=-a$.

Since
\begin{eqnarray}
[\tilde L_{e_\alpha}^2, \tilde Y_u]&=& \{\tilde L_{e_\alpha}, [\tilde L_{e_\alpha}, \tilde Y_u]\}= \{\tilde L_{e_\alpha}, -\tilde Y_{ue_\alpha}\}\cr
&=& -2\tilde Y_{ue_\alpha}\tilde L_{e_\alpha}-[\tilde L_{e_\alpha}, \tilde Y_{ue_\alpha}]\cr
&=& -2\tilde Y_{ue_\alpha}\tilde L_{e_\alpha}+\tilde Y_{(ue_\alpha)e_\alpha},\nonumber
\end{eqnarray}
we have
\begin{eqnarray}
\sum_{0\le \alpha\le D}[\tilde L_{e_\alpha}^2, \tilde Y_u]
&=& \sum_{0\le \alpha\le D}-2\tilde Y_{ue_\alpha}(\hat L_{e_\alpha}-\lambda_{e_\alpha})+\tilde Y_{(ue_\alpha)e_\alpha}\cr
&=& 2i\hat L_{ux}-2i\lambda_{ux}+\rho\left (1+{\rho-2\over 4}\delta\right )\tilde Y_u+{\rho \delta\over 4}\tr u \tilde Y_e\cr
&=& 2i\hat L_{ux}+\rho\left (1+{3\rho-4\over 4}\delta\right )\tilde Y_u+{\rho \delta\over 4}\tr u \tilde Y_e, \nonumber
\end{eqnarray}
here, we have used the identity
$$
\sum_\alpha L_{e_\alpha}^2 =\rho\left (1+{\rho-2\over 4}\delta\right )L_e+{\rho^2 \delta\over 4}\mid e\rangle\langle e\mid.
$$
On the other hand,
\begin{eqnarray}
[\tilde L_e^2+{1\over 2}\{\tilde X_e,\tilde Y_e\}, \tilde Y_u]&=&\{\tilde L_e, [\tilde L_e, \tilde Y_u]\}+{1\over 2}\{[\tilde X_e, \tilde Y_u], \tilde Y_e\}\cr
&=&-2(\tilde Y_u\tilde L_e+\tilde Y_e\tilde L_u-\tilde Y_u), \nonumber
\end{eqnarray}
so we have
\begin{eqnarray}
[\ms O, \tilde Y_u]&=& {4i\over \rho}\hat L_{ux}+2\left (1+{3\rho-4\over 4}\delta\right )\tilde Y_u+{\delta\over 2}\tr u \tilde Y_e\cr
&& +2(\tilde Y_u\hat L_e+\tilde Y_e\hat L_u-\tilde Y_u)-2(\tilde Y_u\lambda_e+\tilde Y_e\lambda_u)\cr
&=& {4i\over \rho}\hat L_{ux}+2(\tilde Y_u\hat L_e+\tilde Y_e\hat L_u)\cr
&=& {2i\over \rho}\sum_{0\le \beta\le D} \hat L_u(\langle x^2-\tr x\, x\mid {/\hskip -5pt
\partial}\rangle)\cr
&=&0 \nonumber
\end{eqnarray}
as differential operator on the Kepler cone.

Since
\begin{eqnarray}
 \sum_{0\le \alpha\le D} \tilde L_{e_\alpha}^2 (1)&=& \sum_{0\le \alpha\le D} -\hat L_{e_\alpha}(\lambda_{e_\alpha})+\lambda_{e_\alpha}^2\cr
 &=& {\rho(\rho-2)\over 4}\delta\lambda_e+{\rho^2\over 8}\delta\lambda_e={\rho(3\rho/2-2)\over 4}\delta\lambda_e, \nonumber
\end{eqnarray}
and
\begin{eqnarray}
 -\tilde L_e^2(1)-{1\over 2}\{\tilde X_e,\tilde Y_e\} (1)=-\lambda_e^2-\lambda_e-B,\nonumber
\end{eqnarray}
we have
\begin{eqnarray}
 -a &= & {(3\rho/2-2)\over 2}\delta\lambda_e-\lambda_e^2-\lambda_e-B\cr
 &=& -{\rho \delta\over 4}(1+{(\rho-2)\delta\over 4}),\nonumber
\end{eqnarray}
i.e., $$a={\rho \delta\over 4}\left(1+{(\rho-2)\delta\over 4}\right).$$
\end{proof}

 \begin{cor}(The Secondary Quadratic Relations)
 \begin{eqnarray}
\sum_{0\le \alpha\le D}\{\tilde X_{e_\alpha}, \tilde L_{e_\alpha}\} =\rho \{
 \tilde X_e, \tilde L_e\},\quad
 \sum_{0\le \alpha\le D}\{\tilde Y_{e_\alpha}, \tilde L_{e_\alpha}\} = \rho\{
 \tilde Y_e, \tilde L_e\}, \cr
\sum_{0\le \alpha\le D}\tilde X_{e_\alpha}^2 = \rho\tilde X_e^2, \quad
\sum_{0\le \alpha\le D}\tilde Y_{e_\alpha}^2 =\rho\tilde Y_e^2,\quad
{1\over 2} \sum_{0\le \alpha\le D} \{\tilde X_{e_\alpha}, \tilde Y_{e_\alpha}\}=\rho(\tilde L_e^2+a),\cr
{2\over \rho}\sum_{1\le \alpha\le D}\{\tilde L_{e_\alpha, u}, \tilde L_{e_\alpha}\}={1\over 2}\left(-\{\tilde X_u, \tilde Y_e\}+\{\tilde X_e, \tilde Y_u\} \right),\cr
{2\over \rho}\sum_{1\le \alpha\le D}\{\tilde L_{e_\alpha, u}, \tilde X_{e_\alpha}\}=-\{\tilde X_u, \tilde L_e\}+\{\tilde L_u, \tilde X_e\},\cr
  {2\over \rho}\sum_{1\le \alpha\le D}\{\tilde L_{e_\alpha,u}, \tilde Y_{e_\alpha}\}=\{\tilde Y_u, \tilde L_e\}-\{\tilde L_u, \tilde Y_e\},\cr
A\sum_{1\le\alpha, \beta\le D}[\tilde L_{e_\alpha}, \tilde L_{e_\beta}]^2={1\over 2}\{\tilde X_e, \tilde Y_e\}-\tilde L_e^2+{\rho \delta\over 4}({\rho \delta\over 4}-1).\nonumber
\end{eqnarray}
Here, $A={2/\rho^2\over 1+{(\rho-2)\delta\over 4}}$. \end{cor}
\begin{proof}
The two identities
\begin{eqnarray}\label{case2}
 \sum_{0\le \alpha\le D}\{\tilde X_{e_\alpha}, \tilde L_{e_\alpha}\} =\rho \{
 \tilde X_e, \tilde L_e\},\quad  \sum_{0\le \alpha\le D}\{\tilde Y_{e_\alpha}, \tilde L_{e_\alpha}\} = \rho\{
 \tilde Y_e, \tilde L_e\}
\end{eqnarray}
can be obtained by taking the commutator of Eq. (\ref{case1}) with $\tilde X_e$ and $\tilde Y_e$ respectively.
Identity
\begin{eqnarray}\label{case3}
{2\over \rho}\sum_{1\le \alpha\le D}\{\tilde L_{e_\alpha, e_\beta}, \tilde L_{e_\alpha}\}+{1\over 2}\left(\{\tilde X_{e_\beta}, \tilde Y_e\}-\{\tilde X_e, \tilde Y_{e_\beta}\} \right)=0
\end{eqnarray}
can be obtained by forming the commutator of the identity in Eq. (\ref{case1}) with $\tilde L_{e_\beta}$ .

Identities
\begin{eqnarray}\label{case4}
\left\{
\begin{array}{rcl}
\sum_{0\le \alpha\le D}\tilde X_{e_\alpha}^2 & = & \rho \tilde X_e^2,\\\\
\sum_{0\le \alpha\le D}\tilde Y_{e_\alpha}^2 &= & \rho \tilde Y_e^2,\\\\
 {2\over \rho} \sum_{0\le
\alpha\le D} \{\tilde X_{e_\alpha}, \tilde Y_{e_\alpha}\} &=& 4(\tilde L_e^2 + a)
\end{array}
\right.
\end{eqnarray}
can be proved this way: The 1st identity here is obtained from forming the commutator of the 1st identity in Eq. (\ref{case2}) with $\tilde X_e$, and the 2nd identity here is obtained from forming the commutator of the 2nd identity in Eq. (\ref{case2}) with $\tilde Y_e$. The third identity here is obtained by first forming the commutator of the 1st identity in Eq. (\ref{case2}) with $\tilde Y_e$ and then using Eq. (\ref{case1}).

Identities
\begin{eqnarray}\left\{\begin{array}{rcl}
{2\over \rho}\sum_{1\le \alpha\le D}\{\tilde L_{e_\alpha, e_\beta}, \tilde X_{e_\alpha}\}+\{\tilde X_{e_\beta}, \tilde L_e\}-\{\tilde L_{e_\beta}, \tilde X_e\}
  & = & 0,\\
 \\
  {2\over \rho}\sum_{1\le \alpha\le D}\{\tilde L_{e_\alpha, e_\beta}, \tilde Y_{e_\alpha}\}-\{\tilde Y_{e_\beta}, \tilde L_e\}+\{\tilde L_{e_\beta}, \tilde Y_e\}
  & = & 0\end{array}\right.
\end{eqnarray}
can be obtained by taking the commutator of Eq. (\ref{case3}) with $\tilde X_e$ and $\tilde Y_e$ respectively.

The last identity in the corollary is a direct consequence of the definition of $X$, but  can be verified to be  a consequence of the TKK commutation relations in Theorem \ref{Main1} and Eq. (\ref{case1}).

\end{proof}

\section{J-Kepler Problems}\label{S:KP}
\begin{Definition}[J-Kepler Problem]\label{D:main2}
The {\bf J-Kepler problem} associated to a simple euclidean Jordan algebra
with rank $\rho$ and degree $\delta$ is the quantum mechanical system for
which the configuration space is the Kepler cone, and the
hamiltonian is
\begin{eqnarray}
\hat h=-{1\over 2}\Delta-\left( {B\over 2\langle e\mid x\rangle^2}+
{1\over \langle e\mid x\rangle}\right).
\end{eqnarray}
Here, $\Delta$ is the (non-positive) Laplace operator on the Kepler
cone, and $$B={\delta(\rho-2)\over 8}\left(\left({3\rho\over
2}-1\right)\delta-2\right).$$
\end{Definition}
\begin{rmk} In view of Eq. (\ref{Xdef}), we have another expression for $\hat h$:
\begin{eqnarray}
\hat h={1\over \langle e\mid x\rangle }\left({i\over 2}\tilde X_e-1\right).
\end{eqnarray}
\end{rmk}
\begin{rmk}
The J-Kepler problem is really the
quantum mechanical system for which the configuration space is the
geometric open cone over the projective space, and the hamiltonian
is
\begin{eqnarray}
\hat h=-{1\over 2}\Delta-\left( {B\over 2r^2}+ {1\over r}\right),
\end{eqnarray}
where $\Delta$ is the (non-positive) Laplace operator on the
open geometric cone over the projective space.
\end{rmk}
\begin{rmk}
The {\bf classical J-Kepler problem} is the classical mechanical system
for which the configuration space is the Kepler cone, and the
Lagrangian is
$$
L(x, \dot x)={1\over 2}|\dot x|^2+ {1\over
\langle e\mid x\rangle}.
$$
\end{rmk}

We are now ready to state
\begin{Prop}The J-Kepler problems are equivalent to the various
Kepler-type problems constructed and analyzed in Ref. \cite{meng},
but with zero magnetic charge. Here is the precise identification:
\begin{displaymath}
\begin{array}{|c|c|c|c|c|c|}
\hline
V & \Gamma(n) & {\mathcal H}_n(\bb R) & {\mathcal H}_n(\bb C) &{\mathcal H}_n(\bb H) & {\mathcal H}_3(\bb O)\\
\hline
& & & & &\\
\begin{matrix}\mbox{Kepler} \cr \mbox{problem}\end{matrix} & \begin{matrix}\mbox{MICZ in} \cr \mbox{dim. $n$}\end{matrix}
&  \begin{matrix}\mbox{$\mr O(1)$ in} \cr \mbox{dim.
$n$}\end{matrix}  &
\begin{matrix}\mbox{$\mr U(1)$ in} \cr \mbox{dim. $(2n-1)$}\end{matrix}
& \begin{matrix}\mbox{$\mr {Sp}(1)$ in} \cr \mbox{dim. $(4n-3)$}\end{matrix}   & \mbox{exceptional}\\
& & & & &\\
\hline
\end{array}
\end{displaymath}
\end{Prop}
\begin{proof}
The proposition is clear for the MICZ-Kepler problems and
the exceptional Kepler problem. For the remaining cases, all one needs
is to make a transformation similar to the one appeared in the proof
of Proposition 2.2 of the first paper in Ref. \cite{meng}. For
example, for the $\mr O(1)$-Kepler problem in dimension $n$ with
zero magnetic charge, the configuration space is $\widetilde{\bb
RP^n}=\bb R^n_*/Z\sim -Z$ and the hamiltonian is
$$H=-{1\over 8z}\Delta_{\widetilde{\bb RP^n}}{1\over z}-{1\over z^2},$$ where
$\Delta_{\widetilde{\bb RP^n}}$ is the Laplace operator on
$\widetilde{\bb RP^n}$ and $z=|Z|$. Note that, with the quotient
metric induced from the euclidean metric of $\bb R^n$,
$\widetilde{\bb RP^n}$ is isometric to
$$\left(\bb R_+\times \bb RP^{n-1}, dz^2+z^2ds^2_{FS}\right),$$
where $ds^2_{FS}$ is the Fubini-Study metric on $\bb RP^{n-1}$.

To see the equivalence of the J-Kepler problem associated with $\mathcal H_n(\bb R)$ with the
$\mr O(1)$-Kepler problem in dimension $n$ with zero magnetic charge, we start with
diffeomorphism
\begin{eqnarray}
\widetilde{\bb RP^n} &\to & {\ms P} \cr [Z] &\mapsto & n ZZ'
\end{eqnarray}
or equivalently diffeomorphism $\pi$:
\begin{eqnarray}
\widetilde{\bb RP^n} &\to & \bb R_+\times \bb P\cr [Z] &\mapsto &
(z^2, \sqrt{2n}{ZZ'\over z^2}).
\end{eqnarray}
Here $Z$ is viewed as a column vector in $\bb R^n$ and $Z'$ is the transpose of $Z$.

Under $\pi$, we have
\begin{eqnarray}
\pi^*(dr^2+r^2\,ds^2_{\bb P}) = (2z)^2
(dz^2+z^2\,ds^2_{FS}),\nonumber
\end{eqnarray}
i.e.,
\begin{eqnarray}
\pi^*(ds^2_{{\ms P}})= (2z)^2ds^2_{\widetilde{\bb RP^n}},
\quad\mbox{so}\quad\pi^*(\mr{vol}_{\ms
P})=(2z)^n\mr{vol}_{\widetilde{\bb RP^n}}.
\end{eqnarray}

Let $\Psi_i$ ($i=1$ or $2$) be a wave-function for the J-Kepler problem associated with $\mathcal H_n(\bb R)$, and
$$\psi_i(z, \Theta):=(2z)^{n\over 2}\pi^*(\Psi_i)(z,
\Theta).$$ Then it is not hard to see that
$$\displaystyle \int_{\widetilde{\bb
RP^n}}\overline{\psi_1}\,\psi_2\,\mr{vol}_{\widetilde{\bb
RP^n}}=\displaystyle \int_{\widetilde{\bb
RP^n}}\overline{\pi^*(\Psi_1)}\pi^*(\Psi_2) \pi^*(\mr{vol}_{\ms
P})=\displaystyle \int_{{\ms P}}\overline{\Psi_1}
\Psi_2\mr{vol}_{{\ms P}}$$ and
\begin{eqnarray}\displaystyle
\displaystyle\int_{\widetilde{\bb RP^n}}\overline{\psi_1}H\psi_2
\mr{vol}_{\widetilde{\bb RP^n}} &=&
\displaystyle\int_{\widetilde{\bb
RP^n}}\overline{\pi^*(\Psi_1)}\,{1\over z^{n\over 2}}Hz^{n\over
2}\,\pi^*(\Psi_2) \,\pi^*(\mr{vol}_{{\ms P}})\cr&=&
\displaystyle\int_{{\ms P}}\overline{\Psi_1}\,{1\over z^{n\over
2}}Hz^{n\over 2}\,\Psi_2 \,\mr{vol}_{{\ms P}}. \nonumber
\end{eqnarray}
Since
\begin{eqnarray}
{1\over z^{n\over 2}}H z^{n\over 2} &=& -{1\over 8 z^{{n\over
2}+1}}\Delta_{\widetilde{\bb RP^n}} z^{{n\over 2}-1}-{1\over z^2}
\cr &=&-{1\over 8 z^{{n\over 2}+1}}\left({1\over
z^{n-1}}\partial_zz^{n-1}\partial_z+{1\over
z^2}\Delta_{FS}\right)z^{{n\over 2}-1}-{1\over z^2}\cr &=& -{1\over
8 z^{3n\over 2}}\partial_zz^{n-1}\partial_zz^{{n\over 2}-1}-{1\over
8z^4}\Delta_{FS}-{1\over z^2}\cr &=&  -{1\over 2 r^{{3n\over
4}-{1\over 2}}}\partial_r r^{n\over 2}\partial_r r^{{n\over
4}-{1\over 2}}-{1\over 2r^2}\Delta_{\bb P}-{1\over r}\cr &=&
-{1\over 2}\left( \partial_r^2+{n-1\over r}\partial_r+{({n\over
4}-{1\over 2})({3n\over 4}-{3\over 2})\over r^2}+{1\over
r^2}\Delta_{\bb P}\right)-{1\over r}\cr &=& -{1\over 2}
\Delta-{3(n-2)^2\over 32r^2}-{1\over r}=\hat h,\nonumber
\end{eqnarray}
we have the equivalence of the J-Kepler problem associated with $\mathcal H_n(\bb R)$
with the $\mr O(1)$-Kepler problem in dimension $n$ with zero
magnetic charge.

\end{proof}

\subsection {The Lenz Vector}
The Lenz vector exists for J-Kepler problems.
\begin{Definition}[Lenz vector]
The Lenz vector for the J-Kepler problem is
$$A_u:={1\over \langle e\mid x\rangle}\left[\tilde L_u, (\langle e\mid x\rangle)^2\hat h\right],$$ i.e.,
\begin{eqnarray}
\fbox{$A_u ={i\over 2}\tilde X_u-\langle u\mid x\rangle \hat h ={i\over 2}\left(\tilde X_u-{\langle u\mid x\rangle\over \langle e\mid x\rangle}\tilde X_e\right)+{\langle u\mid x\rangle\over \langle e\mid x\rangle}$.}
\end{eqnarray}
\end{Definition}

Note that $A_e=1$. One might call $\vec A:=(A_{e_1}, \ldots, A_{e_D})$ the Lenz vector, that is because $\vec A$ is precisely the usual Lenz vector when the Jordan algebra is the Minkowski space. 
\begin{Thm} For $u, v\in V$, we let
$L_{u,v}:=[\hat L_u, \hat L_v]$.  Then we have  commutation relations:
\begin{eqnarray}\label{CTR}
\fbox{$\begin{array}{lcl} [L_{u,v}, \hat h] & = & 0\cr
[L_{u,v},  L_{z, w}] & = &  L_{[L_u, L_v]z,w}+ L_{z, [L_u, L_v]w}\cr
[L_{u,v}, A_z] &= &
A_{[L_u, L_v]z}\cr
[A_u, \hat h] & = & 0\cr[A_u, A_v] & = &
-2\hat h L_{u,v}.
\end{array}$}
\end{eqnarray}
\end{Thm}

\begin{rmk} 1) $L_{u,v}=[\tilde L_u, \tilde L_v]$.  2) $A_u$, $\hat h$ and  $iL_{u,v}$ are all hermitian operators with respect to inner product
$$
(\psi, \phi)\mapsto \int_{\ms P} \psi^* \phi\, \mr{vol}_{\ms P}.
$$
\end{rmk}

\begin{proof}
Theorem \ref{Main1} implies that
$$
[L_{u,v}, \tilde X_z]= \tilde X_{[L_u, L_v]z}, \quad [L_{u,v}, \tilde Y_z]= \tilde Y_{[L_u, L_v]z},\quad [L_{u,v}, \tilde L_z] = \tilde L_{[L_u, L_v]z}.
$$
In particular, we have $[L_{u,v}, \tilde X_e]=[L_{u,v}, \tilde Y_e]=0$. Therefore, $[L_{u, v}, H]=0$,
\begin{eqnarray}
[L_{u, v}, A_z] &=& {1\over \langle e\mid x\rangle}[L_{u,v}, [\tilde L_z, (\langle e\mid x\rangle)^2\hat h]]\cr
&=& {1\over \langle e\mid x\rangle}\left([(\langle e\mid x\rangle)^2\hat h, [\tilde L_z, L_{u,v}]]+[\tilde L_z, [L_{u,v}, (\langle e\mid x\rangle)^2\hat h] ]\right)\cr
&=& {1\over \langle e\mid x\rangle} [(\langle e\mid x\rangle)^2\hat h, -\tilde L_{[L_u, L_v]z}]\cr
&=& A_{[L_u, L_v]z},\nonumber
\end{eqnarray}
and
\begin{eqnarray}
[L_{u, v}, L_{z, w}] &=&  [L_{u,v}, [\tilde L_z, \tilde L_w]]\cr
&=& [[L_{u,v}, \tilde L_z], \tilde L_w]+ [\tilde L_z, [L_{u,v}, \tilde L_w]]\cr
&=& [ \tilde L_{[L_u, L_v]z}, \tilde L_w]+ [\tilde L_z, \tilde L_{[L_u, L_v]w}]]\cr
&=&L_{[L_u, L_v]z, w}+ L_{z,[L_u, L_v]w}.\nonumber
\end{eqnarray}
Since $\hat h={1\over \langle e\mid x\rangle}({i\over 2}\tilde X_e-1)$, we have
\begin{eqnarray}
[A_u, \hat h] &=&  [A_u,{1\over \langle e\mid x\rangle}]({i\over 2}\tilde X_e-1)+{1\over \langle e\mid x\rangle}[A_u, {i\over 2}\tilde X_e]\cr
&=&  -{1\over \langle e\mid x\rangle}[A_u,\langle e\mid x\rangle]{1\over \langle e\mid x\rangle} ({i\over 2}\tilde X_e-1)+{1\over \langle e\mid x\rangle}[A_u, {i\over 2}\tilde X_e]\cr
&=&  -{1\over \langle e\mid x\rangle}[A_u,\langle e\mid x\rangle]\hat h-{1\over \langle e\mid x\rangle}\left[{\langle u\mid x\rangle\over \langle e\mid x\rangle}, {i\over 2}\tilde X_e\right]\langle e\mid x\rangle \hat h\cr
&=&  -{1\over \langle e\mid x\rangle}[A_u,\langle e\mid x\rangle]\hat h\cr
&&-{1\over \langle e\mid x\rangle}\left( [\langle u\mid x\rangle, {i\over 2}\tilde X_e]-{\langle u\mid x\rangle\over \langle e\mid x\rangle}[\langle e\mid x\rangle, {i\over 2}\tilde X_e]\right) \hat h\cr
&=&  -{1\over \langle e\mid x\rangle}[A_u,\langle e\mid x\rangle]\hat h
+{1\over \langle e\mid x\rangle}\left(\tilde L_u-{\langle u\mid x\rangle\over \langle e\mid x\rangle}\tilde L_e\right) \hat h\cr
&=&0.\nonumber
\end{eqnarray}
Since $A_u= {i\over 2}\tilde X_u-\langle u\mid x\rangle \hat h$, we have
\begin{eqnarray}
[A_u, A_v] &=& [{i\over 2}\tilde X_u, -\langle v\mid x\rangle \hat h ]-[{i\over 2}\tilde X_v, -\langle u\mid x\rangle \hat h ]+[\langle u\mid x\rangle \hat h, \langle v\mid x\rangle \hat h]\cr
&=& [{i\over 2}\tilde X_u, -\langle v\mid x\rangle \hat h ]-\langle v\mid x\rangle [\hat h, \langle u\mid x\rangle] \hat h-<u\leftrightarrow v>\cr
&=& [{i\over 2}\tilde X_u, -\langle v\mid x\rangle] \hat h -\langle v\mid x\rangle[{i\over 2}\tilde X_u,  \hat h] -\langle v\mid x\rangle [\hat h, \langle u\mid x\rangle] \hat h\cr && -<u\leftrightarrow v>\cr
&=& -\tilde S_{uv} \hat h +{\langle v\mid x\rangle\over \langle e\mid x\rangle}[{i\over 2}\tilde X_u,  \langle e\mid x\rangle] H-{\langle v\mid x\rangle\over \langle e\mid x\rangle} [{i\over 2}\tilde X_e, \langle u\mid x\rangle] \hat h\cr && -<u\leftrightarrow v>\cr
&=& -\tilde S_{uv} \hat h +{\langle v\mid x\rangle\over \langle e\mid x\rangle}\tilde L_u \hat h-{\langle v\mid x\rangle\over \langle e\mid x\rangle}\tilde L_u \hat h -<u\leftrightarrow v>\cr
&=& -2L_{u,v} \hat h =-2\hat h L_{u,v}.\nonumber
\end{eqnarray}

\end{proof}

\section{Symmetry Analysis of the J-Kepler Problems}\label{S:Summary}
The goal of this section is to give a detailed dynamical symmetry
analysis for the J-Kepler problem, as a byproduct, we solve the
bound state problem for the J-Kepler problem algebraically.

Unless said otherwise, throughout this section we assume that  $V$ is a simple
euclidean Jordan algebra with rank $\rho\ge 2$ and degree $\delta$. For
simplicity, for each $x\in {\ms P}$, we shall rewrite $\langle e\mid
x\rangle$ as $r$.

\subsection{Harmonic Analysis on Projective Spaces}
Let us begin with the harmonic analysis on projective
space $\bb P$. Since $\bb P$ is a real affine variety inside $V$,
its coordinate ring ${\bb R}[{\bb P}]$ is a quotient of the ring of
real polynomial functions on $V$. Recall that we use $\Delta_{\bb
P}$ to denote the Laplace operator on $\bb P$.
\begin{Lem}\label{Lemma8}
Let $V_m$ be the set of regular (i.e., polynomial) functions (on $\bb P$) of degree at
most $m$, ${\mathcal V}_m$ be the orthogonal complement of $V_{m-1}$
in $V_m$. For each $u\in V$, we let $m_u$:
${\bb R}[{\bb P}]\to {\bb R}[{\bb P}]$ denote the multiplication by
$\langle u\mid x\rangle$.

i) For any integer $k\ge 0$, there is a $u\in V$ perpendicular to $e$ such that
$$
\tilde m_u: \;\mathcal V_k\buildrel m_u\over \to V_{k+1}\buildrel \pi\over \to \mathcal V_{k+1}
$$
is nonzero. Here, $\pi$ is the orthogonal projection.

ii) ${\mathcal V}_m^{\bb C}:={\mathcal V}_m\otimes_{\bb R} \bb C$ is the $m$-th eigenspace of $\Delta_{\bb P}$
with eigenvalue $-m(m+{\rho \delta\over 2}-1)$, and the Hilbert space of
square integrable complex-valued functions on $\bb P$ admits an
orthogonal decomposition into the eigenspaces of $\Delta_{\bb P}$:
\begin{eqnarray}\label{decom}
L^2({\bb P})=\hat \bigoplus_{m\ge 0} {\mathcal V}_m^{\bb C}.
\end{eqnarray}

\end{Lem}
\begin{proof} It is clear that
$$V_m=\bigoplus_{k=0}^m{\mathcal V}_k, \quad {\bb R}[{\bb P}]=\bigoplus_{k=0}^\infty{\mathcal
V}_k.$$

i) Let $\{e\}\cup\{e_i\mid 1\le i\le \dim V-1\}$ be an orthonomal
basis for $V$ and $x_i=\langle
e_i\mid x\rangle$. Since $m_u$ maps $V_k$ to $V_{k+1}$ for each
integer $k\ge 0$, we have the resulting map
$$
\overline{m_u}:\; V_k/V_{k-1}\to  V_{k+1}/V_{k}.
$$
For any $\bar f\in V_{k+1}/V_{k}$, if we write $\bar f=\sum_{i>0}
\overline{x_ig_i}$, we have $\bar f=\sum_{i>0}
\overline{m_{e_i}}(\overline{g_i})$. In view of the commutative
diagram
$$\begin{CD}
{\mathcal V}_k @>\tilde m_u>> {\mathcal V}_{k+1}\\
@VV \cong V @VV\cong  V\\
V_k/V_{k-1} @>\overline{m_u}>> V_{k+1}/V_{k}\; ,
\end{CD}
$$ we have
\begin{eqnarray}\label{nonzero}
{\mathcal V}_{k+1}=\sum_{i>0} \tilde m_{e_i}({\mathcal V}_{k})
\end{eqnarray} for any $k\ge 0$.
Suppose that $\tilde m_{e_i}$: ${\mathcal V}_m \to {\mathcal
V}_{m+1}$ is zero for any $i>0$, then ${\mathcal V}_{m+1}=0$ per Eq. (\ref{nonzero}), so ${\mathcal V}_n=0$ for any $n\ge m+1$ per Eq. (\ref{nonzero}), then $${\bb R}[{\bb P}]=\lim_{k\to\infty} V_k=V_m,$$ a contradiction.

ii) Let $u_1$, \ldots, $u_m$ be in $V$, and write $x_{u_i}$ for
$\langle u_i\mid x\rangle$, $u_i^0$ for $\langle e\mid u_i\rangle$. Since
\begin{eqnarray}
[\Delta_{\bb P}, x_{u_1}\cdots x_{u_m}] &=& \sum_{i=1}^m
x_{u_1}\cdots x_{u_{i-1}}[\Delta_{\bb P},x_{u_i}]x_{u_{i+1}}\cdots
x_{u_m}, \nonumber
\end{eqnarray}
in view of Eq. (\ref{identity}) and part ii) of Lemma \ref{lemma}, we have
\begin{eqnarray}
\Delta_{\bb P}\left( x_{u_1}\cdots x_{u_m}\right) &=& \sum_{i=1}^m
x_{u_1}\cdots x_{u_{i-1}}[\Delta_{\bb P},x_{u_i}](x_{u_{i+1}}\cdots
x_{u_m})\cr &=& \sum_{i=1}^m x_{u_1}\cdots x_{u_{i-1}}(-2r\tilde
L_{u_i}+2x_{u_i}\tilde L_e)(x_{u_{i+1}}\cdots x_{u_m})\cr &=&
-m{\rho \delta\over 2}x_{u_1}\cdots x_{u_m}+{\rho \delta\over
2}ru_i^0 \sum_{i=1}^m x_{u_1}\cdots \hat x_{u_i}\cdots x_{u_m}\cr
&&+\sum_{i=1}^m x_{u_1}\cdots x_{u_{i-1}}(-2r\hat
L_{u_i}+2x_{u_i}\hat L_e)(x_{u_{i+1}}\cdots x_{u_m})\cr &=& -(m{\rho
\delta\over 2}+m(m-1))x_{u_1}\cdots x_{u_m}+{\rho \delta\over
2}r\sum_{i=1}^m u_i^0 x_{u_1}\cdots \hat x_{u_i}\cdots x_{u_m}\cr
&&-2r\sum_{i=1}^m x_{u_1}\cdots
x_{u_{i-1}}\hat L_{u_i}(x_{u_{i+1}}\cdots x_{u_m})\cr &\equiv& -m(m+{\rho
\delta\over 2}-1)x_{u_1}\cdots x_{u_m}\quad (\mod V_{m-1})\nonumber
\end{eqnarray} because $r=\sqrt{2/\rho}$ on $\bb P$.

It is then clear that $\Delta_{\bb P}$ maps $V_m$ into $V_m$ for
each $m\ge 0$ and resulting map
$$\overline{\Delta_{\bb P}}:\; V_m/V_{m-1}\to V_m/V_{m-1}$$
is the scalar multiplication by $-m(m+{\rho \delta\over 2}-1)$. Since
$\Delta_{\bb P}$ is a hermitian operator and maps $V_{m-1}$ into $V_{m-1}$,
$\Delta_{\bb P}$ maps ${\mathcal V}_m$ into ${\mathcal V}_m$, so we have
commutative diagram
$$\begin{CD}
{\mathcal V}_m @>\Delta_{\bb P}>> {\mathcal V}_m\\
@VV\cong V @VV\cong V\\
V_m/V_{m-1} @>\overline{\Delta_{\bb P}}>> V_m/V_{m-1}\; .
\end{CD}
$$
Since ${\mathcal V}_m\neq\{0\}$ per part i), we conclude that
${\mathcal V}_m$ is an eigenspace of $\Delta_{\bb P}$ with
eigenvalue $-m(m+{\rho \delta\over 2}-1)$.

Since the ring of regular functions is dense in the ring of real
continuous functions and
$$
{\bb R}[{\bb P}]=\bigoplus_{k=0}^\infty{\mathcal
V}_k,
$$
we have
$$
L^2({\bb P})=\hat \bigoplus_{m\ge 0} {\mathcal V}_m^{\bb C}.
$$

\end{proof}

\subsection{Associated Lagueree Polynomials} Here, we give a quick review of the associated Lagueree
polynomials. Let $\alpha$ be a real number and $n\ge 0$ be an
integer. By definition, the associated Lagueree polynomial
$L^\alpha_n(x)$ is the polynomial solution of equation
\begin{eqnarray}\label{LaguereeEq}
xy''+(\alpha+1-x)y'+ny=0
\end{eqnarray} whose the leading
coefficient is $(-1)^n{1\over n!}$. We note that $L^\alpha_n(x)$ has
degree $n$, for example,
\begin{eqnarray}
L_0^\alpha(x) &=& 1\cr L_1^\alpha(x) &=& -x+\alpha+1\cr
L_2^\alpha(x)&=&{1\over 2}x^2-(\alpha+2)x+{1\over
2}(\alpha+1)(\alpha+2)\cr
 & \vdots&\nonumber
\end{eqnarray}
In general, we have
\begin{eqnarray}
L_n^\alpha(x)={x^{-\alpha}e^x\over n!}{d^n\over
dx^n}\left(e^{-x}x^{n+\alpha}\right).
\end{eqnarray}
It is a fact that $L^\alpha_n(x)$'s form an orthogonal basis for
$L^2({\bb R}_+, x^\alpha e^{-x}\,dx)$:
\begin{eqnarray}
\displaystyle \int _0^\infty x^\alpha e^{-x}\,
L_n^\alpha(x)L_m^\alpha(x)\, dx={\Gamma(n+\alpha+1)\over
n!}\delta_{mn};
\end{eqnarray}moreover, a degree $n$ polynomial in $x$ can be
uniquely written as a linear combination of $L_k^\alpha(x)$ with
$0\le k\le n$.

 It is then clear that,
\begin{eqnarray}\label{Leguree}
\fbox{$\begin{array}{c} \mbox{\em for any integer $l\ge 0$,
$x^{l-{(\rho/2-1)\delta\over 2}}e^{-x}L^{2l+{\rho \delta\over 2}-1}_n(2x)$'s
form}\\ \text{\em an orthogonal basis for $L^2({\bb R}_+,
x^{(\rho-1)\delta-1}\,dx)$}.\end{array}$}
\end{eqnarray}
It is also a fact that
\begin{eqnarray}\label{recursive1}
nL_n^\alpha(x)=(n+\alpha)L_{n-1}^\alpha(x)-xL_{n-1}^{\alpha+1}(x)
\end{eqnarray}
and \begin{eqnarray}\label{recursive2} L_{n+1}^\alpha(x)={1\over
n+1}\left((2n+1+\alpha-x)L_n^\alpha(x)-(n+\alpha)L_{n-1}^\alpha(x)\right).
\end{eqnarray}

\subsection{Hidden Harmonic Analysis on Kepler Cones}
For each integer $l\ge 0$, we fix an orthonormal spanning set
$\{Y_{lm}\mid m\in{\mathcal I}(l)\}$ for ${\mathcal V}_l$. Note that
each $Y_{lm}$ can be represented by a homogeneous degree
$l$-polynomial in $x$, which will be denoted by $Y_{lm}(x)$. For integer $k\ge 1$, we introduce
\begin{eqnarray}\label{wavefunction}
\varphi_{klm}(x): = r^{-{(\rho/2-1)\delta\over 2}}L_{k-1}^{2l+{\rho
\delta\over 2}-1}(2r) e^{-r} Y_{l m}(x)
\end{eqnarray} where $r=\langle e\mid x\rangle$. One can verify that
$\varphi_{klm}$ is square integrable with respect to ${1\over r}\mr{vol}_{\ms P}$:
\begin{eqnarray} \left({2\over
\rho}\right)^l\displaystyle\int_{\ms P}
|\varphi_{klm}|^2\,{1\over r}\mr{vol}_{\ms P}&=& \int_{\bb P}|Y_{lm}|^2\,\mr{vol}_{\bb P}\;\cdot \cr && \cdot
\int_0^\infty r^{2l-(\rho/2-1)\delta}\cdot(L_{k-1}^{2l+{\rho \delta\over
2}-1}(2r))^2\cdot r^{(\rho-1)\delta-1}e^{-2r}\, dr\cr
&=&\int_0^\infty r^{2l+{\rho \delta\over 2}-1}(L_{k-1}^{2l+{\rho
\delta\over 2}-1}(2r))^2e^{-2r}\, dr\cr &= & {\Gamma(2l+{\rho \delta\over
2}-1+k)\over 2^{2l+{\rho \delta/2}}(k-1)!}<\infty\nonumber
\end{eqnarray} because $\delta\ge 1$ and $\rho\ge 1$.

Let $H_0:=-{i\over 2}(X_e+Y_e)$, then
$$
\tilde H_0={1\over 2}\left({\hat L_e^2-(2\lambda_e-1)\hat
L_e+\Delta_{\bb P}+B\over r}-r\right).
$$

We say that a smooth nonzero function $\varphi$ on the Kepler cone
is an {\bf eigenfunction} of $\tilde H_0$ if it is square integrable
with respect to ${1\over r}\mr{vol}_{\ms P}$ and satisfies equation
$$
\tilde H_0\varphi=\lambda \varphi
$$ for some real number $\lambda$. With the help of part ii) of Lemma \ref{Lemma8} and Eq. (\ref{LaguereeEq}), one can check that
$$\tilde H_0\varphi_{klm} = -(l+k-1+{\rho \delta\over 4})\varphi_{klm},$$
so {\em $\varphi_{klm}$ is an eigenfunction of $\tilde H_0$ with eigenvalue
$-(l+k-1+{\rho \delta\over 4})$}.

Let ${\ms V}_l(k):=\mr{span}_{\bb C}\{\varphi_{klm}\mid m\in {\mathcal I}(l)\}$ for each integer $l\ge 0$ and
$k\ge 1$, and
$$
\tilde {\ms H}_I:=\bigoplus_{l=0}^{I}{\ms V}_l(I+1-l)
$$for each integer $I\ge 0$. Finally, we let $\tilde {\mathcal H}:=\bigoplus_{I=0}^\infty \tilde {\ms H}_I$ and
$\pi(\mathcal O):=\tilde{\mathcal O}$ for any $\mathcal O$ in the
conformal algebra. In view of statement
(\ref{Leguree}), it is clear that $\varphi_{klm}$'s form an orthogonal basis for $\tilde {\mathcal H}$.
\begin{Prop}

i) $\tilde {\ms H}_I$ is the eigenspace of $\tilde H_0$ with eigenvalue
$-(I+\rho \delta/4)$ and
\begin{eqnarray} L^2({\ms P}, {1\over r}\mr{vol}_{\ms P})=
\hat \bigoplus_{I=0}^\infty \tilde {\ms H}_I.\nonumber
\end{eqnarray}
Moreover, $\varphi_{klm}$'s form an orthogonal basis for $L^2({\ms
P}, {1\over r}\mr{vol}_{\ms P})$.

ii) $(\pi,\tilde {\mathcal H})$ is a unitary representation of the conformal
algebra.

iii)  $(\pi|_{\bar{\frk u}}, \tilde {\ms H}_I)$ is an irreducible
representation of $\bar{\frk u}$. Consequently $(\pi|_{{\frk u}},
\tilde {\ms H}_I)$ is an irreducible representation of ${\frk u}$.

iv) $(\pi,\tilde {\mathcal H})$ is a unitary lowest weight representation of the
conformal algebra with lowest weight equal to ${\delta \over
2}\lambda_0$. Here $\lambda_0$ is the fundamental weight conjugate
to the unique non-compact simple root $\alpha_0$ in Lemma
\ref{LemmaVogan}, i.e., $\lambda_0(H_{\alpha_i})=0$ for $i>0$ and
$$
2{\lambda_0(H_{\alpha_0})\over \alpha_0(H_{\alpha_0})}=1.
$$

\end{Prop}
\begin{proof}
i) By virtue of Theorem II.10 of Ref. \cite{Reed&Simon}, we have
$$L^2({\ms P}, {1\over r}\mr{vol}_{\ms P})=L^2(\bb R_+,
r^{(\rho-1)\delta-1}\,dr)\otimes L^2({\bb P}).$$  Then
\begin{eqnarray}
L^2({\ms P}, {1\over r}\mr{vol}_{\ms P}) & = & \hat
\bigoplus_{l=0}^\infty \left(L^2(\bb R_+,
r^{(\rho-1)\delta-1}\,dr)\otimes {\mathcal V}_l^{\bb C}\right)\quad \text{Eq.
(\ref{decom})}\cr &=& \hat \bigoplus_{l=0}^\infty \hat
\bigoplus_{k=1}^\infty \bigoplus_{m\in {\mathcal I}(l)}
\mr{span}_{\bb C}\{\varphi_{klm}\}\quad\text{statement
(\ref{Leguree})}\cr &=& \hat \bigoplus_{I=0}^\infty \tilde {\ms
H}_I.\nonumber
\end{eqnarray}
Therefore, we conclude that $\{\varphi_{klm}\}$ is an orthogonal
basis for $L^2({\ms P}, {1\over r}\mr{vol}_{\ms P})$ and $\tilde {\ms H_I}$
is the $I$-th eigenspace of $\tilde H_0$ with eigenvalue $-(I+\rho \delta/4)$.

ii) First, we need to show that $\tilde{\mathcal O}(\psi)\in \tilde {\mathcal
H}$ for any $\mathcal O\in \frk{co}$ and any $\psi\in \tilde {\mathcal H}$.
Without loss of generality we may assume that $\mathcal O$ is $L_u$,
$X_e$ or $Y_e$ and
$$\psi=\varphi_{klm}=r^{-{(\rho/2-1)\delta\over 2}}L_{k}^{2l+{\rho \delta\over
2}-1}(2r) e^{-r} Y_{l m}(x).$$ Using Eq. (\ref{recursive2}) one can
see that $\tilde Y_e(\varphi_{klm})\in\tilde {\mathcal H}$. Then
\begin{eqnarray}
{\tilde X_e}(\varphi_{klm}) &= &(2\sqrt{-1}\tilde H_0-\tilde
Y_e)(\varphi_{klm})\cr
&= & 2\sqrt{-1}(k+l-1+{\rho \delta\over 2})\varphi_{klm}-\tilde
Y_e(\varphi_{klm})\cr
&\in &\tilde {\mathcal H}.\nonumber
\end{eqnarray}
Since
\begin{eqnarray}
\hat L_u(\varphi_{klm}) &= &\hat L_u(r^{-{(\rho/2-1)\delta\over
2}}L_{k}^{2l+{\rho \delta\over 2}-1}(2r) e^{-r}) Y_{l m}(x)\cr &&
+r^{-{(\rho/2-1)\delta\over 2}}L_{k}^{2l+{\rho \delta\over 2}-1}(2r) e^{-r}
\hat L_u(Y_{l m}(x)), \nonumber
\end{eqnarray}
one can see that $\hat L_u(\varphi_{klm})\in \bigoplus_{k'\le k+l+1, l'\le
l+1}{\ms V}_{l'}(k')$. It is also clear that $$\lambda_u\cdot \varphi_{klm}\in
\bigoplus_{k'\le k+l+1, l'\le l+1}{\ms V}_{l'}(k'),$$
so $\tilde L_u(\varphi_{klm})\in \tilde {\mathcal H}$.

Next, we verify that
\begin{eqnarray}
(\varphi_{klm},\,\tilde{\mathcal O}(\varphi_{k'l'm'}))+(\tilde
{\mathcal O}(\varphi_{klm}),\,\varphi_{k'l'm'})=0
\end{eqnarray}for $\mathcal O\in \frk{co}$. We may assume that $\mathcal O$ is $L_u$, $X_e$ or
$Y_e$. It is clearly OK when $\mathcal O=Y_e$ because $\tilde
Y_e=-ir$. Since $\tilde X_e=2i \tilde H_0-\tilde Y_e$, to show that $\tilde X_e$ is anti-hermitian, it suffices to verify that $(\varphi_{klm},\,\tilde
H_0(\varphi_{k'l'm'}))-(\tilde
H_0(\varphi_{klm}),\,\varphi_{k'l'm'})=0$ or
$(k'+l'-k-l)(\varphi_{klm},\,\varphi_{k'l'm'})=0$, which is
obviously true.

To verify that $(\varphi_{klm},\,\tilde
L_u(\varphi_{k'l'm'}))+(\tilde
L_u(\varphi_{klm}),\,\varphi_{k'l'm'})=0$, in view of part iii) of
Lemma \ref{KeyLemma}, we know that $(\varphi_{klm},\,\tilde
L_u(\varphi_{k'l'm'}))+(\tilde
L_u(\varphi_{klm}),\,\varphi_{k'l'm'})$ is equal to
$$ \displaystyle\int_{{\ms P}}{\ms L}_u(\overline{\varphi_{klm}}\,\varphi_{k'l'm'}\,
{1\over r}\mr{vol}_{{\ms P}})=\displaystyle\int_{\ms P}d \iota_{\hat
L_u}(\overline{\varphi_{klm}}\,\varphi_{k'l'm'}\, {1\over
r}\mr{vol}_{\ms P}) =0.$$ That is because $\iota_{\hat
L_u}(\overline{\varphi_{klm}}\,\varphi_{k'l'm'}\, {1\over
r}\mr{vol}_{\ms P})$ approaches to zero exponentially fast as $r\to
\infty$ and approaches to zero as $r\to 0$, uniformly with respect
to the angle directions.

iii) First, we verify that $\tilde {\ms H}_I$ is invariant under the
action of $\bar{\frk u}$. To see this, we note that $\tilde {\ms H}_I$ is an
eigenspace of $\tilde H_0$, moreover, as operators on Hilbert space
$L^2({\ms P}, {1\over r}\mr{vol}_{\ms P})$, $\tilde H_0$ commutes with $\tilde {\mathcal O}$
for any $\mathcal O\in \bar{\frk u}$.

Since $\tilde {\ms H}_I=\bigoplus_{l=0}^I{\ms V}_l(I+1-l)$, if the action
of $\bar{\frk u}$ on $\tilde {\ms H}_I$ were not irreducible, there would be
an integer $l$ with $0\le l< I$ such that $(\psi_l, \tilde {\mathcal
O}(\psi_{l+1}))=0$ for any $\psi_l\in {\ms V}_l(I+1-l)$,
$\psi_{l+1}\in {\ms V}_{l+1}(I-l)$, and any $\mathcal O \in \bar{\frk
k}$.

In view of part i) of Lemma \ref{Lemma8}, we can choose a $u\in V$
with $u\perp e$ such that $\tilde m_u$: ${\mathcal V}_l\to {\mathcal
V}_{l+1}$ is nontrivial; so there is a $Y_{l m}\in {\mathcal V}_l$ and a
$Y_{(l+1) m'}\in {\mathcal V}_{l+1}$ such that
\begin{eqnarray}\label{notzero}
\int_{\bb P}Y_{(l+1)m'}\cdot  \tilde m_u(Y_{lm})\, \mr{vol}_{\bb P}\neq 0.
\end{eqnarray}
Let
\begin{eqnarray}
\psi_l(x)&=& r^{-{(\rho/2-1)\delta\over 2}}L_{k}^{2l+{\rho \delta\over
2}-1}(2r) e^{-r} Y_{l m}(x),\cr \psi_{l+1}(x)&=&
r^{-{(\rho/2-1)\delta\over 2}}L_{k-1}^{2l+{\rho \delta\over 2}+1}(2r) e^{-r}
Y_{(l+1) m'}(x),\cr
{\mathcal O} &=& X_u+Y_u.\nonumber
\end{eqnarray}

Then $\mathcal O\in \bar{\frk u}$ because $u\perp e$. Since
$\tilde {\mathcal O}=[\tilde L_u, \tilde X_e-\tilde Y_e]=[2i\tilde
L_u, \tilde H_0]+2\tilde Y_u$, we have $(\psi_l, \tilde {\mathcal
O}(\psi_{l+1}))=(\psi_l, 2\tilde Y_u \cdot \psi_{l+1})$; so, in
view of Eq. (\ref{notzero}), $(\psi_l, \tilde {\mathcal
O}(\psi_{l+1}))=0$ would imply that
\begin{eqnarray}
\int_0^\infty x^{\alpha+2} e^{-x}\, L_k^{\alpha}(x)
L_{k-1}^{\alpha+2}(x)\, dx =0. \nonumber\end{eqnarray} where $\alpha=2l+\rho
\delta/2 -1$. But that is a contradiction: using Eq. (\ref{recursive1}), one can show that
\begin{eqnarray}
\int_0^\infty x^{\alpha+2} e^{-x}\, L_k^{\alpha}(x)
L_{k-1}^{\alpha+2}(x)\, dx = -2{\Gamma(k+\alpha+2)\over (k-1)!}\neq
0.\nonumber
\end{eqnarray}

iv) Let us take the simple root system $\alpha_0$, \ldots, $\alpha_r$
specified in Lemma \ref{LemmaVogan} and $$\psi_0(x) = r^{-{(\rho/2-1)\delta\over 2}}e^{-r}.$$

Since $\tilde {\ms H}_0$ ($=\mr{span}_{\bb C}\{\psi_0\}$) is one dimensional
and $\bar{ \frk u}$ is semi-simple, the action of  $\bar{ \frk u}^{\bb C}$ on $\tilde {\ms H}_0$ must be
trivial. Therefore, for $i\ge 1$,  in view of the fact that $E_{\pm\alpha_i}, H_{\alpha_i}\in
\bar{ \frk u}^{\bb C}$, we have
\begin{eqnarray}\label{lw1}
\tilde E_{-\alpha_i}\psi_0=0, \quad \tilde H_{\alpha_i} \psi_0=0.
\end{eqnarray}
On the other hand, since $E_{-\alpha_0}={i\over
2}(X_{e_{11}}-Y_{e_{11}})+L_{e_{11}}$ and
$H_{\alpha_0}=i(X_{e_{11}}+Y_{e_{11}})\equiv -{2\over \rho}H_0(\mod \bar {\frk u})$, by a computation, we have
\begin{eqnarray}\label{lw2}
\tilde E_{-\alpha_0}\psi_0=0, \quad \tilde H_{\alpha_0}\psi_0={\delta\over 2}\psi_0.
\end{eqnarray}
Therefore,  in view of the fact that $\alpha_0(H_{\alpha_0})=2$, {\em $\psi_0$ is a lowest weight state with weight ${\delta\over 2}\lambda_0$}.

Since operator $\tilde Y_v$ is the multiplication by $-i\langle v\mid x\rangle$, we have
$$
r^{-{(\rho/2-1)\delta\over 2}}e^{-r}\sum_{i_1, \ldots, i_n} \alpha_{i_1\cdots i_n} x_1^{i_1}\cdots x_n^{i_n}=\left(\sum_{i_1, \ldots, i_n} \alpha_{i_1\cdots i_n} (i\tilde Y_{e_1})^{i_1}\cdots(i\tilde Y_{e_n})^{i_n}\right) \psi_0,
$$
so the representation $(\pi, \tilde {\ms H})$ is generated from $\psi_0$. Since this representation is unitary, it must be irreducible. 

In summary, $(\pi, \tilde {\ms H})$ is a unitary lowest weight representation with lowest weight ${\delta\over 2}\lambda_0$.

\end{proof}

\begin{rmk}
Let $\mathcal P$ be the space of regular functions on $\ms P$, i.e.,
$$
{\mathcal P}=\{p:{\ms P}\to {\bb C}\mid \mbox{$p$ is a polynomial on $V$} \}.
$$
In view of Eqn. (\ref{wavefunction}), 
$$
{\tilde D}:=e^{-r} r^{-{(\rho/2-1)d\over 2}}{\mathcal P}
$$
can be taken as a dense common domain of definition for $\tilde {\mathcal O}$, $\mathcal O\in \frk{co}$.
\end{rmk}

The following main theorem is an easy corollary of the above proposition.

\begin{Thm}\label{main2}
Let $\mr{Co}$ be the conformal group of the Jordan algebra with rank at least two, $\mr K$
be the closed Lie subgroup of $\mr{Co}$ whose Lie algebra is $\frk u$,
$\lambda_0$ be the fundamental weight conjugate to the unique non-compact simple root $\alpha_0$ in Lemma
\ref{LemmaVogan}, $\tilde H_0:=-{i\over 2}(\tilde X_e+\tilde Y_e)$, $\tilde {\ms H}_I$ be
the $I$-th eigenspace of $\tilde  H_0$, and $\tilde {\mathcal
H}:=\bigoplus _{I=0}^\infty\tilde {\ms H}_I$.

1) The hidden action $\pi$ in Theorem \ref{Main1} turns $\tilde{\mathcal H}$
into a unitary lowest weight $(\frk{co}, \mr{K})$-module with
lowest weight ${\delta\over 2}\lambda_0$. Here the action is unitary
with respect to inner product
$$
(\psi_1, \psi_2)=\displaystyle\int_{\ms
P}\overline\psi_1\,\psi_2\,{1\over r}\mr{vol}_{\ms P}\; .
$$

2) The unitary lowest weight representation of $\mr{Co}$, whose underlying $(\frk{co}, \mr{K})$-module is
the $(\frk{co}, \mr{K})$-module in part 1), can be realized by $L^2({\ms P}, {1\over r}\mr{vol}_{\ms P})$.

3) Decomposition $\tilde {\mathcal H}=\bigoplus _{I=0}^\infty\tilde {\ms H}_I$ is
a multiplicity free $K$-type formula.
\end{Thm}
Note that, the unitary lowest weight representation of $\mr{Co}$ appeared in this theorem is the minimal representation of $\mr{Co}$ in the sense of A. Joseph \cite{Joseph1974}, and has the smallest positive Gelfand-Kirillov dimension. This theorem has a more general version which takes care of
all unitary lowest weight representations of the smallest positive Gelfand-Kirillov dimension. Since it is a refinement of part (ii) of Theorem XIII.3.4 from Ref. \cite{FK91} for the case $\nu={\delta\over 2}$ there, this theorem can be conceivably generalized to cover the case for a generic $\nu$ there.

\subsection{Solution of the J-Kepler Problems}
For a J-Kepler problem, we are primarily interested in solving
the bound state problem here, i.e., the following (energy) spectrum problem:
\begin{eqnarray}\label{eigen}
\left\{\begin{array}{rcl}
\hat h\psi & = & E\psi\\
\\
\displaystyle\int_{{\ms P}} |\psi|^2\, \mr{vol}_{\ms P}&< & \infty, \quad \psi\not\equiv 0.
\end{array}\right.
\end{eqnarray}
It turns out that $E$ has to take ceratin discrete values. For
example, for the original Kepler problem, we have
$$
E=-{1\over 2n^2}, \quad n=1, 2, \ldots
$$
The {\bf Hilbert space of bound states}, denoted by $\ms H$, is
defined to be the completion of the linear span of all
eigenfunctions of $\hat h$.

\begin{Thm}\label{main3} Let $V$ be a simple euclidean Jordan
algebra with rank $\rho\ge 2$ and degree $\delta$ and $\mr{Co}$ be the conformal group of $V$. For the J-Kepler problem associated to $V$, the following statements are true:

1) The bound state energy spectrum is
$$
E_I=-{1/2\over (I+{\rho \delta\over 4})^2}
$$ where $I=0$, $1$, $2$, \ldots

2) There is a unitary action of $\mr{Co}$ on the Hilbert space of bound states, ${\ms
H}$. In fact, ${\ms H}$ provides a realization for
the minimal representation of the conformal
group $\mr{Co}$.

3) The orthogonal decomposition of $\ms H$ into the energy
eigenspaces is just the multiplicity free $K$-type formula for the
minimal representation.

\end{Thm}
\begin{proof}
We start with the eigenvalue problem for $\tilde  H_0$:
\begin{eqnarray}
\tilde H_0\tilde \psi=-n_I \tilde \psi
\end{eqnarray} where $n_I=(I+\rho \delta/4)$ and $\tilde\psi$ is square integrable
with respect to measure ${1\over r}\mr{vol}_{\ms P}$ and
$\tilde\psi\not\equiv 0$. The above equation can be recast as
$$
-{1\over 2}\left(\Delta+{B\over r^2}+{2n_I\over
r}\right)\tilde\psi(x) =-{1\over 2}\tilde \psi(x).
$$
Let $\psi(x):=\tilde \psi({x\over n_I})$, then the preceding
equation becomes
$$
\left(-{1\over 2}\Delta-{B\over 2r^2}-{1\over r}
\right)\psi(x)=-{1/2\over n_I^2}\psi(x),$$
i.e.,
\begin{eqnarray}
\hat h \psi =-{1/2\over n_I^2}\psi.
\end{eqnarray}
One can check that $\psi$ is square integrable
with respect to measure $\mr{vol}_{\ms P}$. Therefore, $\tilde\psi$ is an eigenfunction of $\tilde H_0$
$\Rightarrow$ $\psi$ is an eigenfunction of $\hat h$. By turning the above arguments backward, one can show that
the converse of this statement is also true. Therefore,

\begin{eqnarray}
\fbox{$\tilde\psi$ is an eigenfunction of $\tilde H_0$
$\Leftrightarrow$ $\psi$ is an eigenfunction of $\hat h$.}
\end{eqnarray}

Introduce $${\ms H}_I:=\{\psi\mid \tilde \psi\in \tilde{\ms H}_I\}, \quad {\mathcal H}:=\bigoplus_{i=0}^\infty {\ms H}_I,$$
and denote by $\tau$: ${\mathcal H}\to \tilde {\mathcal H}$ the linear map such that
$$\fbox{$\tau(\psi)(x)=n_I^{{(\rho-1)\delta\over 2}+1}\psi(n_Ix)$}$$
for $\psi\in {\ms H}_I$. By virtue of Theorem 2 in Ref. \cite{CD2003}, one can show that $\tau$ is an isometry. Here, the inner product on $\mathcal H$ is
the usual one: for $\psi$, $\phi$ in $\mathcal H$, we have
$$
\langle \psi, \phi\rangle=\int_{\ms P}\overline{\psi}\,\phi\, \mr{vol}_{\ms P}.
$$

Since $\tilde{\mathcal H}$ is a unitary lowest weight Harish-Chandra module, and $\tau$ is an isometry,
${\mathcal H}$ becomes a unitary lowest weight Harish-Chandra module. Since the completion of $\mathcal H$ is the
Hilbert space of bound states, the rest is clear from Theorem \ref{main2}.
\end{proof}

\appendix
\section{Proof of $O(1)=0$ and ${\ms O}(1)=0$}\label{S:Boring}
We shall write
$$
X=-{1\over r}(R+\Delta+B)
$$ where $R=\hat L_e^2-((\rho-1)\delta -1)\hat L_e$ and $\Delta=A\sum_{\alpha,\beta}[\hat L_{e_\alpha}, \hat L_{e_\beta}]^2$. The following facts shall be used often in the
computations:
\begin{eqnarray}
[\Delta, \langle u\mid x\rangle] & = & -2r\tilde L_u+2\langle
u\mid x\rangle\tilde L_e, \cr \Delta (\langle u\mid x\rangle) & = &
-{\rho \delta\over 2}(\langle u\mid x\rangle-\langle u\mid e\rangle r),\cr R(f(x)) &= &
(k^2+k-(\rho-1)\delta k)f(x)\quad \mbox{if $f(x)$ has homogeneous degree
$-k$.}\nonumber
\end{eqnarray}
Since
\begin{eqnarray}
[\tilde L_u, X]=-{\langle u\mid x\rangle \over
r^2}(R+\Delta+B)-{1\over r}[\hat L_u, \Delta]-{1\over
r}[\Delta, \lambda_u],
\end{eqnarray}
we have
\begin{eqnarray}
[\tilde L_u, X](1)&=& -B\langle u\mid x\rangle/r^2-\Delta
(\lambda_u)/r\cr &=& -B\langle u\mid
x\rangle/r^2+{(\rho/2-1)\delta\over 2 r^2}[-\Delta, \langle u\mid
x\rangle]\cr &=& \left(-B+{\rho(\rho-2)\delta^2\over 8}\right){\langle
x\mid u\rangle \over r^2}-{\rho(\rho-2)\delta^2\over 8}{\langle u\mid e\rangle\over
r}.
\end{eqnarray}

We may further assume that $\langle u\mid e\rangle=\langle v\mid e\rangle=0$ in the computations below.

{\bf Proof of $O(1)=0$}. Since
\begin{eqnarray}
[[\tilde L_u, X], \lambda_v](1) &=& [[\tilde L_u, \lambda_v],
X](1)+[[\lambda_v, X], \tilde L_u](1)\cr &=& [\hat L_u(\lambda_v),
X](1)+[{1\over r}[\Delta, \lambda_v], \tilde L_u](1)\cr &=&
{1\over r}[\Delta,\hat L_u (\lambda_v)](1)+(\rho/2-1)\delta[-{1\over
r}\tilde L_v+{\langle x\mid v\rangle\over r}\tilde L_e, \tilde
L_u](1)\cr &=& {1\over r}\Delta(\hat L_u (\lambda_v))\cr
&&+(\rho/2-1)\delta\left([{-1\over r}, \tilde L_u](-\lambda_v)+[\tilde
L_u, {\langle x\mid v\rangle\over r^2}](\lambda_e)\right)\cr &=
&{(\rho -2)\delta\over 4r}\Delta\left(-{\langle x\mid uv\rangle \over
(r)}+{\langle x\mid u\rangle \langle x\mid v\rangle\over
r^2}\right)\cr &&-{(\rho-2)(\rho-1)\delta^2\over 4}{\langle x\mid
uv\rangle\over r^2}\cr &&+(\rho/2-1)\delta(3\rho/4-1/2)\delta{\langle
x\mid u\rangle\langle x\mid v\rangle\over r^2}\cr &=& {(\rho
-2)\delta\over 4r}\left({\rho \delta\over 2}{\langle x\mid uv\rangle
-r\langle u\mid v\rangle\over r}+{1\over r^2}\Delta(\langle
x\mid u\rangle \langle x\mid v\rangle)\right)\cr &&
-{(\rho-2)(\rho-1)\delta^2\over 4}{\langle x\mid uv\rangle\over
r^2}\cr&& +(\rho/2-1)\delta(3\rho/4-1/2)\delta{\langle x\mid
u\rangle\langle x\mid v\rangle\over r^2}\cr &=& {(\rho -2)\delta\over
4r^3}\left(({\rho \delta\over 2}+2)r\langle x\mid uv\rangle-(\rho
d+2)\langle x\mid u\rangle \langle x\mid v\rangle\right)\cr &&
-{\rho(\rho-2)\delta^2\over 8}{\langle u\mid v\rangle\over
r}-{(\rho-2)(\rho-1)\delta^2\over 4}{\langle x\mid uv\rangle\over
r^2}\cr && +(\rho/2-1)\delta(3\rho/4-1/2)d{\langle x\mid
u\rangle\langle x\mid v\rangle\over r^2}\cr &=& {(\rho-2)\delta\over
4}\left(({\rho\over 2}-1)\delta-2\right) \left(-{\langle x\mid
uv\rangle\over r^2} +{\langle x\mid u\rangle\langle x\mid
v\rangle\over r^2}\right)\cr &&-{\rho(\rho-2)\delta^2\over 8}{\langle
u\mid v\rangle\over r},\nonumber
\end{eqnarray}
and
$$
\hat L_v([\tilde L_u, X](1)) =  \left(-B+{\rho(\rho-2)\delta^2\over
8}\right)\cdot \left[-{\langle x\mid uv\rangle \over
r^2}+2{\langle x\mid u\rangle \langle x\mid v\rangle\over
r^3}\right],
$$
we have
\begin{eqnarray}
[\tilde L_v, [\tilde L_u, X]](1) &=& \tilde L_v ([\tilde L_u,
X](1))+[\tilde L_u, X](\lambda_v)\cr &=& \hat L_v ([\tilde L_u,
X](1))+[[\tilde L_u, X], \lambda_v](1)\cr &=&
\left((B-{\rho(\rho-2)\delta^2\over 8})-{(\rho-2)\delta\over
4}\left(({\rho\over 2}-1)\delta-2\right)\right){\langle x\mid uv\rangle
\over r^2}\cr &&-{\rho(\rho-2)\delta^2\over 8}{\langle u\mid
v\rangle\over r}\cr &&+\left(-2B+{(\rho-2)\delta\over
4}\left(({3\rho\over 2}-1)\delta-2\right)\right){\langle x\mid u\rangle
\langle x\mid v\rangle\over r^3}.\nonumber
\end{eqnarray}
On the other hand,
\begin{eqnarray}
[\tilde L_{uv}, X](1)= \left(-B+{\rho(\rho-2)\delta^2\over
8}\right){\langle x\mid uv\rangle \over
r^2}-{\rho(\rho-2)\delta^2\over 8}{\langle u\mid v\rangle\over
r},\nonumber
\end{eqnarray}
so $[\tilde L_v, [\tilde L_u, X]](1)-[\tilde L_{uv}, X](1)$ is equal
to
\begin{eqnarray}
&&\left(2(B-{\rho(\rho-2)\delta^2\over 8})-{(\rho-2)\delta\over
4}\left(({\rho\over 2}-1)\delta-2\right)\right){\langle x\mid uv\rangle
\over r^2}\cr&&+\left(-2B+{(\rho-2)\delta\over 4}\left(({3\rho\over
2}-1)\delta-2\right)\right){\langle x\mid u\rangle \langle x\mid
v\rangle\over r^3} \cr & = &\left(-2B+{(\rho-2)\delta\over
4}\left(({3\rho\over 2}-1)\delta-2\right)\right)\left(-{\langle x\mid
uv\rangle \over r^2}+{\langle x\mid u\rangle \langle x\mid
v\rangle\over r^3}\right).\nonumber
\end{eqnarray}
Then $O(1)=0$ if and only if
$$
B={\delta\over 8}(\rho-2)\left(\left({3\rho\over 2}-1\right)\delta-2\right).
$$
{\bf Proof of $\ms O(1)=0$}. Since
$$
X([\tilde L_v, X](1))=
(2-(3\rho/2-1)\delta+B)\left(B-{\rho(\rho-2)\delta^2\over 8}\right){\langle
x\mid v\rangle \over r^3},
$$
and
\begin{eqnarray}[\tilde L_u, X](X(1))={\langle u\mid x\rangle
\over r^3}\left[B(2-(\rho-1)\delta+B)+{B\rho \delta\over
2}\left(1-{(\rho-2)\delta\over 4}\right)\right],\nonumber
\end{eqnarray}
we have
$$
[[\tilde L_u, X], X](1)=\rho \delta\left(B-{d\over
8}(\rho-2)\left(\left({3\rho\over
2}-1\right)d-2\right)\right){\langle u\mid x\rangle \over r^3},
$$
so $\ms O(1)=0$ if and only if
$$
B={\delta\over 8}(\rho-2)\left(\left({3\rho\over 2}-1\right)\delta-2\right).
$$

\end{document}